\documentclass[journal]{IEEEtran} 
\usepackage{caption}
\usepackage{xurl}
\usepackage{cite}
\usepackage{amsmath,amssymb,amsfonts}

\usepackage{graphicx}
\usepackage{enumitem}
\usepackage{subcaption}
\usepackage{amsthm}
\usepackage{textcomp}
\usepackage{algorithm}
\usepackage{algpseudocode}  
\usepackage{float}
\usepackage{xspace}
\usepackage[table]{xcolor}
\usepackage{booktabs}
\definecolor{codehl}{RGB}{255,245,200}

\usepackage{bm}
\makeatletter
\AtBeginDocument{\DeclareMathVersion{bold}
\DeclareSymbolFont{NewLetters}{T1}{times}{b}{it}
\SetSymbolFont{operators}{bold}{T1}{times}{b}{n}
\SetSymbolFont{NewLetters}{bold}{T1}{times}{b}{it}
\SetMathAlphabet{\mathrm}{bold}{T1}{times}{b}{n}
\SetMathAlphabet{\mathit}{bold}{T1}{times}{b}{it}
\SetMathAlphabet{\mathbf}{bold}{T1}{times}{b}{n}
\SetMathAlphabet{\mathtt}{bold}{OT1}{pcr}{b}{n}
\SetSymbolFont{symbols}{bold}{OMS}{cmsy}{b}{n}
\renewcommand\boldmath{\@nomath\boldmath\mathversion{bold}}}
\makeatother

\def\BibTeX{{\rm B\kern-.05em{\sc i\kern-.025em b}\kern-.08em
    T\kern-.1667em\lower.7ex\hbox{E}\kern-.125emX}}

\algnewcommand{\Notation}{\item[\textbf{Notation:}]}
\algnewcommand{\Output}{\item[\textbf{Output:}]}

\newcommand\ignore[1]{}
\newcommand{\rtsmarm}{SMARM+\xspace}
\newcommand{\rtsmarmone}{SMARM+PRNG\xspace}
\newcommand{\rtsmarmtwo}{SMARM+FPE\xspace}

\newcommand{\acron}{SMARM+\xspace}
\newcommand{\smarm}{SMARM\xspace}

\newtheorem{definition}{Definition}
\newtheorem{assumption}{Assumption}[section]
\newtheorem{lemma}{Lemma}[section]
\newtheorem{theorem}{Theorem}[section]

\newcommand{\Adv}{\mathbf{Adv}}

\begin{document}

\title{SMARM+: Analyzing and Enhancing Shuffled Measurements for Remote Attestation in Real-Time IoT Settings}
\author{
\IEEEauthorblockN{Amarin Laohajirapan}
\IEEEauthorblockA{\textit{College of Computing,} Prince of Songkla University, Phuket Campus, Thailand\\
Email: 6730621007@psu.ac.th}
\and
\IEEEauthorblockN{Norrathep Rattanavipanon\thanks{Corresponding author: norrathep.r@psu.ac.th}}
\IEEEauthorblockA{\textit{College of Computing,} Prince of Songkla University, Phuket Campus, Thailand\\
Email: norrathep.r@psu.ac.th}
}

\maketitle

\begin{abstract}
Remote attestation (RA) is a lightweight security primitive for detecting software compromise on IoT devices. Traditional RA schemes require atomic, non-interruptible memory measurements, making them difficult to deploy alongside real-time workloads. SMARM addresses this limitation by measuring memory in a secret, shuffled block order, reducing the non-interruptibility period to the duration of a single block measurement. However, SMARM was originally designed for microkernel-based systems and has not been studied in RTOS-driven real-time environments.

In this work, we present the first systematic study of SMARM in real-time RTOS-based setups. We implement SMARM on commodity ARM TrustZone-M hardware running FreeRTOS and Zephyr, and introduce the \emph{Frequency Accuracy Ratio} (FAR) to quantify the extent to which attestation can coexist with real-time execution under varying workloads. Our evaluation shows that SMARM’s real-time compatibility is highly sensitive to block size: large blocks significantly degrade real-time availability, while small blocks incur substantial secure-storage overhead, limiting deployability on memory-constrained devices.

To address this limitation, we propose SMARM+, a family of enhanced SMARM variants consisting of SMARM+PRNG and SMARM+FPE. They are designed to reduce secure-storage requirements while preserving SMARM’s security guarantees and real-time behavior. Our evaluation highlights the trade-off between secure-storage reduction, attestation runtime, and energy overhead, and provides guidance on selecting among SMARM, SMARM+PRNG, and SMARM+FPE for different deployment settings.

\end{abstract}



\maketitle

\section{Introduction}
\label{introduction}

Internet-of-Things (IoT) devices have become integral to modern society, widely deployed in many domains, such as smart healthcare, industrial automation, and smart infrastructure.
As their usage continues to grow (approaching 20 billion devices worldwide as of 2025~\cite{statistaConnectionsWorldwide}), they have also become attractive targets for large-scale malware attacks~\cite{kasperskyKasperskyUnveils}.
However, due to the resource-constrained nature, deploying sophisticated security mechanisms (e.g., anti-virus) is often impractical in IoT devices.

To address this challenge, Remote Attestation (RA) has emerged as a lightweight security primitive for detecting software compromise on such devices.
At a high level, it is a challenge-response protocol between a constrained IoT device, referred to as the \emph{prover}, and a more powerful entity acting as the \emph{verifier}.
The goal of RA is for the prover to convince the verifier that it is currently running an expected software configuration, and thus is free from a malware compromise.
As illustrated in Figure~\ref{fig:ra-overview}, a typical RA protocol consists of three steps:

\begin{figure}[!htp]
    \centering
    \includegraphics[width=\linewidth]{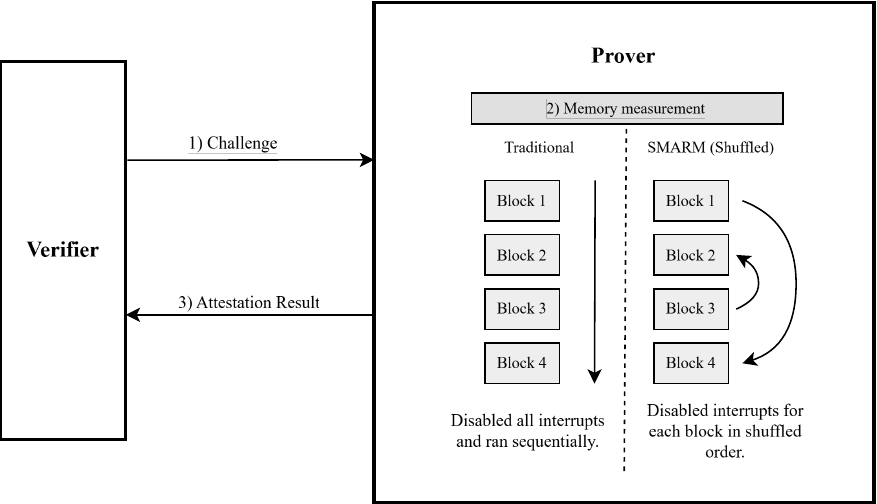}
    \caption{Remote Attestation (RA) workflow and memory measurement strategies.}
    \label{fig:ra-overview}
\end{figure}

\begin{enumerate}[label={(\arabic*)}]
    \item The verifier sends a fresh challenge (e.g., a random nonce or a monotonically increasing counter) to the prover.
    \item Upon receiving the challenge, the prover computes an authenticated integrity measurement -- realized as either a message authentication code (MAC) or a digital signature -- over its program memory along with the received challenge, and returns the result as attestation evidence to the verifier.
    \item The verifier checks the evidence against a known reference value corresponding to a valid software state; a successful match assures the verifier that the prover is not infected by malware.
\end{enumerate}

To prevent \emph{roving} malware~\cite{smarm} that relocates itself to evade detection, traditional RA techniques~\cite{smart,hydra,vrased} require the computation in Step (2) to execute atomically, i.e., with interrupts disabled.
While effective, this approach also prevents normal applications from executing in Step (2), rendering such schemes unsuitable for real-time systems.

SMARM~\cite{smarm} was proposed as an alternative memory measurement strategy to relax the atomicity requirement.
As illustrated in Figure~\ref{fig:ra-overview}, rather than measuring program memory sequentially with interrupts disabled, SMARM allows interrupts to be re-enabled after each memory block is measured.
On its own, this relaxation would permit roving malware to interrupt the measurement task in Step (2) and relocate itself to evade detection.
Hence, SMARM proposes an additional technique to mitigate this by measuring the memory blocks in a \emph{private} and \emph{random} (i.e., shuffled) order.
To be secure, SMARM safeguards the shuffled order in a secure storage that cannot be accessed by untrusted software.
Intuitively, since the roving malware does not know which memory blocks have already been measured, its probability of successfully evading detection is significantly reduced.
With this technique, SMARM reduces the non-interruptibility period from the entire program memory to a single-block measurement, substantially improving real-time availability of RA.

In this work, we observe two limitations of SMARM.
First, SMARM was designed and prototyped atop the HYDRA architecture~\cite{hydra}, which relies on a formally verified seL4 microkernel~\cite{sel4} to secure the measurement task and store the shuffled order.
While this design provides strong security guarantees due to the formal verification of seL4, real-time IoT deployments commonly rely on lightweight RTOSes such as FreeRTOS, rather than microkernel-based systems.
As such, it limits the practicality and wide-scale adoption of SMARM in RTOS-based systems.
Second, SMARM was evaluated in experimental settings that do not consider real-time workloads.
Although the results show a significant reduction in the non-interruptibility period, it remains unclear whether SMARM can truly co-exist with real-time systems.

In light of these limitations, this work makes the following contributions:

\begin{itemize}\sloppy
    \item \textbf{Design and implementation of SMARM on RTOS-based platforms (Section~\ref{sec:porting}).}
    We present a design and implementation of SMARM that is compatible with widely-adopted RTOSes.
    The key challenge is to preserve SMARM's security guarantees without relying on a formally verified microkernel, while maintaining compatibility with RTOS-based execution environments.
    To resolve this challenge, our design leverages the ARM TrustZone-M security extension as a trust anchor in place of seL4.
    TrustZone-M is increasingly available on modern MCU platforms~\cite{armCommunity}, making our solution practical and readily deployable on off-the-shelf real-time IoT devices.
    
    \item \textbf{Evaluation of SMARM under real-time workloads (Section~\ref{sec:evaluation}).}
    We present the first systematic evaluation of SMARM in real-time environments using two widely deployed RTOSes, FreeRTOS and Zephyr.
    To quantify real-time coexistence during attestation, we introduce a new metric, \emph{Frequency Accuracy Ratio (FAR)}.
    Our results show that SMARM's real-time compatibility is highly sensitive to the measurement block size: small blocks preserve near-ideal task execution, whereas larger blocks significantly degrade real-time availability.
    On the other hand, achieving real-time compatibility requires smaller blocks, which substantially increase SMARM's secure-storage requirement, exposing a key limitation for deployment on highly resource-constrained MCUs.
 
    \item \textbf{Addressing SMARM limitation (Section~\ref{sec:resmarm-design}-\ref{sec:rtsmarm-evaluation}).}
    We propose \acron, enhanced variants of SMARM that reduce secure-storage requirements while preserving real-time compatibility and security guarantees.
    The first variant (\rtsmarmone) reduces storage from $O(n \log n)$ to $O(n)$ bits while the second (\rtsmarmtwo) further reduces storage to $O(1)$.
    Our evaluation shows that both variants incur different performance overheads, demonstrating a trade-off between secure-memory usage, attestation latency, and energy consumption. 

\end{itemize}

\section{Background}\label{sec:background}
This section reviews the background of SMARM and ARM TrustZone-M architecture used as basis in this work.

\subsection{SMARM: Shuffled Measurements Against Roving Malware}\label{subsec:smarm}

Remote attestation (RA) for low-end embedded devices has traditionally
relied on \emph{atomic} memory measurements, as in SMART, HYDRA,
and VRASED, where the entire attested memory region is hashed in a
single non-interruptible sweep with interrupts disabled~\cite{smart,hydra,vrased,ra_survey}. These designs provide strong security guarantees but inherently conflict with real-time workloads, since they block all interrupts for a period that scales with the size of the attested memory.
The key idea of SMARM~\cite{smarm} is to enhance the RA workflow (see Figure~\ref{fig:ra-overview}) by measuring memory in a shuffled (i.e., unpredictable and private) order. Specifically, SMARM partitions the prover's memory into $n$ fixed-size blocks,
$M = \{M_1, \ldots, M_n\}$, which are measured according to a one-time shuffled order rather than a fixed/predictable sequence.

Figure~\ref{fig:smarm_state} illustrates the SMARM workflow.
Upon receiving a one-time challenge, $chal$, from the verifier, a secure environment on the prover computes an attestation response via the following steps:

\begin{figure}[!htp]
    \centering
    \includegraphics[width=\linewidth]{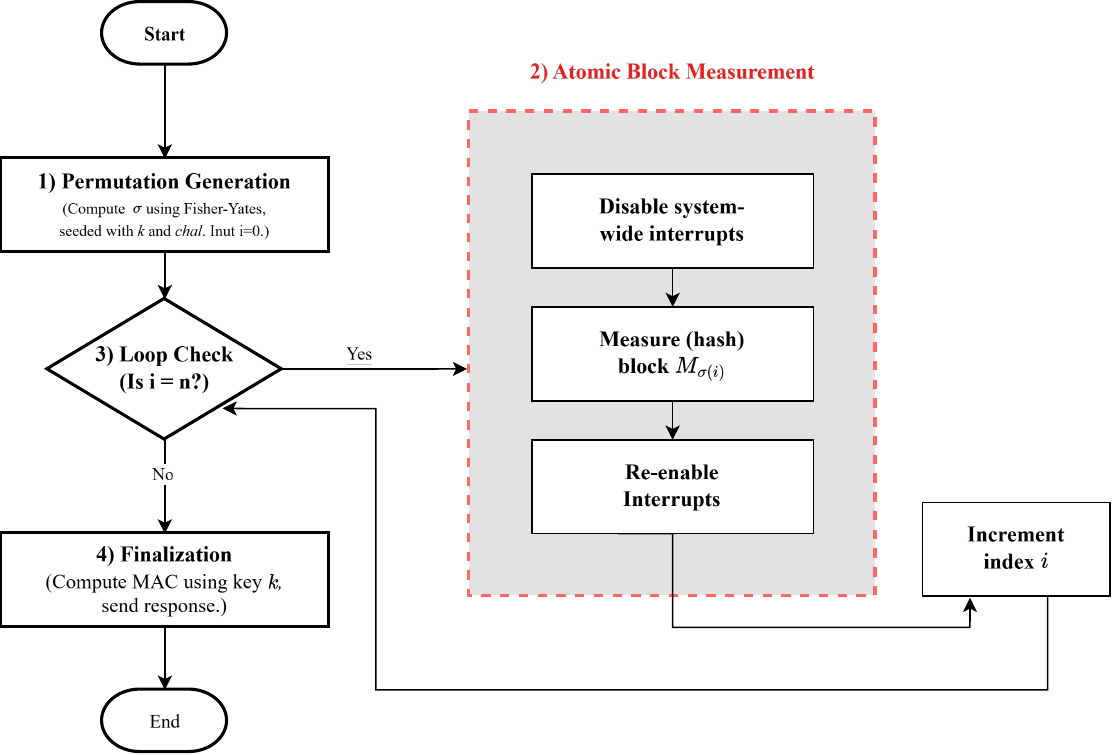}
    \caption{The workflow of SMARM~\cite{smarm}}
    \label{fig:smarm_state}
\end{figure}

\begin{enumerate}
    \item \textbf{Permutation Generation.}
    The prover computes a permutation $\sigma$ over the memory block indices $[1,\ldots,n]$, where $\sigma$ determines the measurement order for the current attestation instance, i.e., the block $M_{\sigma(i)}$ is measured at step $i$.
    SMARM's main security requires $\sigma$ to be \emph{fresh} and \emph{private} for each attestation instance;
    to achieve these, SMARM instantiates this permutation using the Fisher-Yates shuffle method~\cite{fisherYates}:
    \[
        \sigma = \mathsf{FisherYates}(seed, [0,\ldots,n-1]),
    \]
    where its PRNG $seed$ is derived from $k$ and $chal$: $k$ is a pre-shared secret key between the prover and the verifier, and $chal$ is the verifier’s challenge.
    After this, it stores $\sigma$ in secure storage and initializes the loop counter $i$ to $0$.

    \item \textbf{Atomic Block Measurement.}
    At iteration $i$, the prover's secure environment retrieves $\sigma(i)$ from secure storage, disables system-wide interrupts, and measures (i.e., hashes) the corresponding memory block $M_{\sigma(i)}$.
    After the measurement, it re-enables interrupts and increments $i$.

    \item \textbf{Loop Check.}
    The secure environment checks whether all memory blocks have been measured by testing whether $i = n$.
    If not, it repeats the previous step.

    \item \textbf{Finalization.}
    After all blocks in $M$ have been measured, the prover's secure environment computes a MAC over the final measurement using the key $k$ and sends the result to the verifier as the attestation response.
    With possession of $k$ and $chal$, the verifier can reconstruct the permutation used in the current instance and thereby validate the received attestation response against the expected measurement of~$M$.
\end{enumerate}

\paragraph{\textbf{Security}}
SMARM considers a threat model in which the adversary can exploit software vulnerabilities on the prover to gain control over all software except for the secure execution environment and secure storage.
In particular, the adversary is allowed to read/write and execute arbitrary code anywhere in the prover's memory.
In addition, SMARM assumes a \emph{Knowledge of Future Volume} (KFV) malware, which may use the aforementioned capabilities to interrupt the secure execution after each block measurement and learn \emph{how many} memory blocks have already been measured.

Under this threat model, SMARM provides a theoretical analysis of its defense.
Assuming the worst-case scenario in which the malware fits entirely within a single memory block, the optimal evasion strategy for KFV malware is to interrupt after each block measurement and relocate itself to a different block.
For large $n$, this strategy yields an evasion probability converged
to approximately $37\%$~\cite{smarm}.
To increase the likelihood of detection, SMARM proposes repeating
the shuffled measurements multiple times for each attestation instance;
for example, performing the shuffled measurements seven times reduces
the evasion probability to below $10^{-3}$~\cite{smarm}.

To meet security, SMARM assumes architectural support for a secure execution environment that provides the following security guarantees:
\begin{itemize}
    \item[(S1)] \emph{Controlled execution:} the attestation code can only be invoked and return via well-defined entry and exit points (i.e., no mid-code entry or early exit), and it may only be interrupted between consecutive block measurements;
    \item[(S2)] \emph{Key secrecy:} the attestation key~$k$ is accessible exclusively to the attestation code; and
    \item[(S3)] \emph{Private measurement order:} the order in which memory blocks are measured must not be known by KFV malware. This is realized by keeping the permutation~$\sigma$ in memory that is inaccessible to untrusted software.
\end{itemize}

\begin{figure}[!htp]
    \centering
    \includegraphics[width=\linewidth]{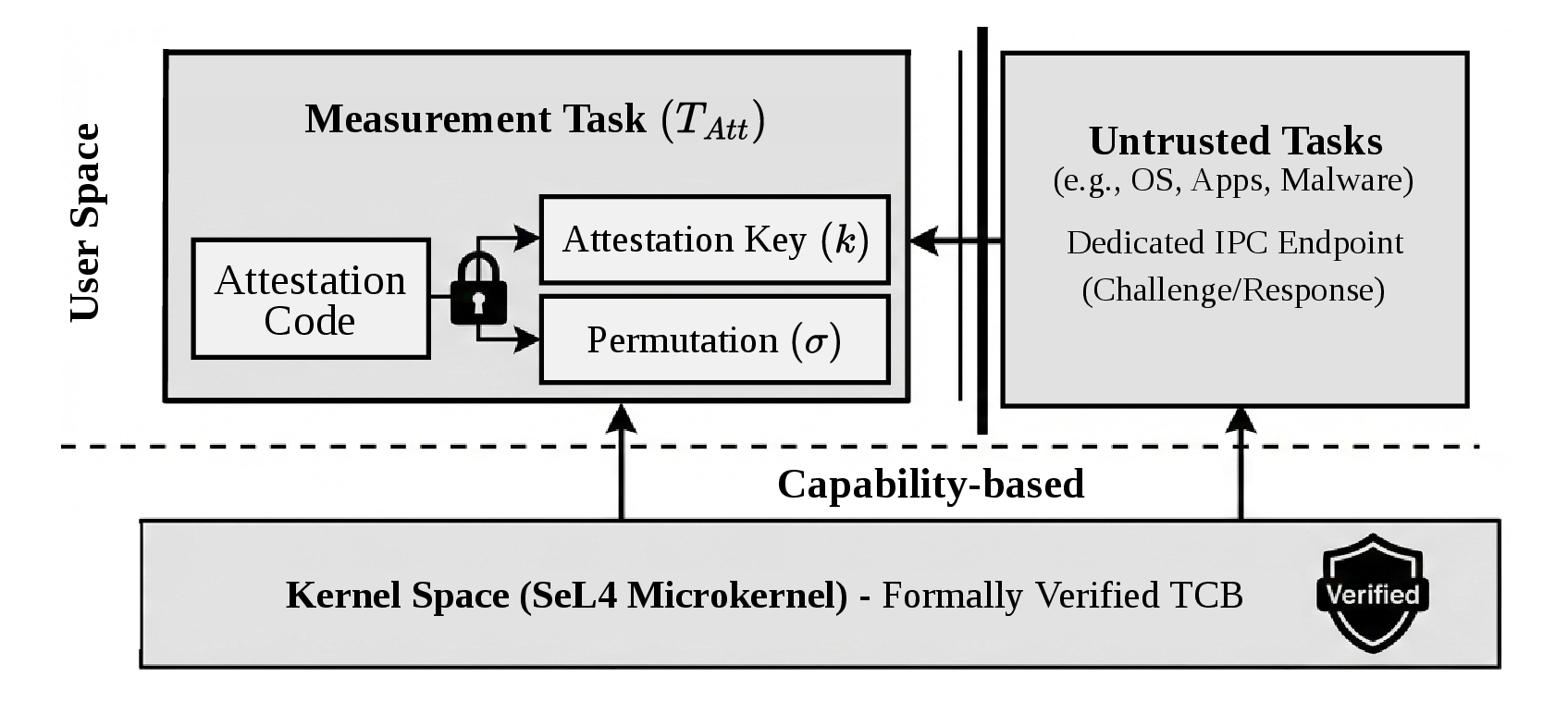}
    \caption{SMARM prototype atop the HYDRA architecture.}
    \label{fig:hydra_impl}
\end{figure}

\paragraph{\textbf{Implementation}}
In its prototype (shown in Figure~\ref{fig:hydra_impl}), SMARM realizes this secure execution environment using the HYDRA architecture~\cite{hydra}, which targets MMU-equipped devices and is built atop the formally verified seL4 microkernel~\cite{sel4}.
Leveraging seL4's verified isolation properties, HYDRA executes the attestation code within a privileged user-space task, referred to as the \emph{measurement task}.
This task is isolated from all other user-space tasks, can access all device memory required for measurement while its own memory remains inaccessible to other tasks.

To enforce (S1), the measurement task can only be invoked via a dedicated IPC endpoint and returns solely with the attestation response.
At runtime, the measurement task is initially assigned the highest priority, ensuring that it cannot be preempted while measuring a memory block.
After completing each block measurement, the task lowers its priority and explicitly yields the CPU execution, allowing other tasks (malicious or real-time tasks) to execute between consecutive measurements.
When the measurement task is scheduled again, it restores its priority to the highest level and proceeds with the next atomic block measurement.
Finally, HYDRA configures access control of memory regions in seL4 such that only the measurement task can access $k$ and $\sigma$, thereby satisfying (S2) and (S3).

\paragraph{\textbf{Evaluation}}
The security of SMARM fundamentally relies on the secrecy of the permutation~$\sigma$, which necessitates secure storage that is inaccessible to untrusted software.
For a prover memory $M$ partitioned into $n$ blocks, SMARM requires at least
$|\sigma| = n \cdot \lceil \log_2 n \rceil$ bits of secure storage to store $\sigma$, making the overhead linear in~$n$.

Reducing the secure storage requirement therefore necessitates lowering the value of $n$.
However, as each block has size $B = |M|/n$, a smaller $n$ implies a larger block size to be measured atomically.
This, in turn, increases the non-interruptibility period $t_{\text{disabled}}$ during attestation in SMARM.

As a result, SMARM exhibits an inherent trade-off between secure storage size and real-time availability.
SMARM evaluates this trade-off on the i.MX6 SabreLite platform, a development board representative of low-cost IoT devices\cite{smarm,sabreLite}.
Figure~\ref{fig:smarm_tdisabled} illustrates the tension: increasing the block size $B$ (equivalent to reducing $n$) alleviates the secure storage requirement, but it also increases $t_{\text{disabled}}$.
Conversely, selecting a small $B$ minimizes $t_{\text{disabled}}$ (improving real-time availability during attestation) at the cost of a much larger secure storage overhead, which may be impractical for resource-constrained devices.

\begin{figure}[!h]
    \centering
    \includegraphics[width=\linewidth]{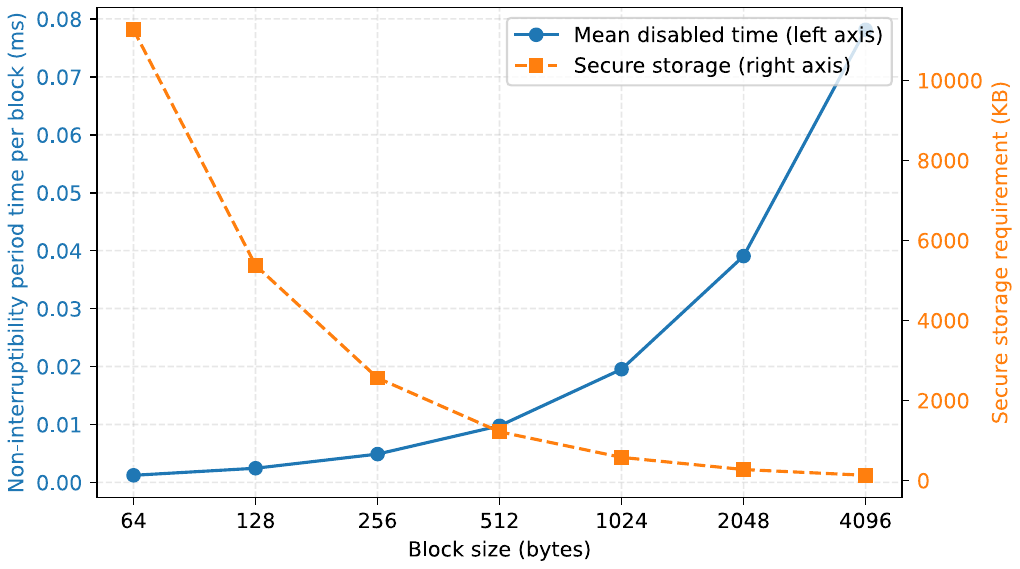}
    \caption{SMARM trade-off between non-interruptibility period $t_{\text{disabled}}$
    and secure storage requirement $|\sigma|$
    for a fixed memory size $|M| = 256\,\text{MB}$ on the i.MX6 SabreLite device.}
    \label{fig:smarm_tdisabled}
\end{figure}

In addition, SMARM performed an additional experiment that evaluates the overall attestation runtime under different values of two parameters: $t_{\text{disabled}}$ and $t_{\text{slice}}$, where the latter represents a time-slice duration that a normal application task runs between consecutive block measurements.
The results showed that both $t_{\text{disabled}}$ and $t_{\text{slice}}$ directly affect the total attestation time with their high value leading to a longer attestation time.
Nonetheless, the work in~\cite{smarm} \emph{did not} evaluate the extent to which the SMARM technique impacts the real-time guarantees of time-sensitive tasks, leaving it unclear which block size should be used in a given strict real-time setting.

\subsection{TrustZone-M Architecture}\label{subsec:tzm}

ARM TrustZone for Cortex-M~\cite{arm-trustzone-cm} (or TrustZone-M in short) provides a production-ready Trusted Execution Environment (TEE) for low-cost IoT devices.
TrustZone-M is commercially available on Cortex-M microcontrollers, which are estimated to be deployed in over 22 billion devices worldwide~\cite{armCommunity}.
It provides hardware-enforced isolation by partitioning software execution into two isolated runtime environments: the \emph{Secure} and \emph{Non-Secure} worlds.

Following the TEE's design principle, TrustZone-M runs untrusted application software (including RTOS and user applications) in the Non-Secure world, while placing small security-critical software in the Secure world.
Although isolated, both worlds still execute on the same physical CPU core, sharing the system's interrupts;
these interrupts are globally managed by the Nested Vectored Interrupt Controller (NVIC), allowing interrupts to happen cross-world boundaries.
As a result, execution in the Secure world may be interrupted by the Non-Secure world; also it is possible for the Secure world to disable all sources of interrupts (e.g., via the \texttt{CPSID} instruction), thereby preventing the Non-Secure world from executing and enabling atomic execution of a Secure-world function.

In this work, we rely on the following security properties offered by TrustZone-M:

\begin{itemize}
    \item \textbf{Isolation at Hardware Level.}
    TrustZone-M provides strong separation between the Secure and Non-Secure worlds through dedicated hardware components: the Security Attribution Unit (SAU) and the Implementation-Defined Attribution Unit (IDAU), which jointly define the security attribution of memory and peripherals at the hardware level~\cite{arm-trustzone-cm}.
    These components govern access to all system resources including program/data memory and peripherals.
    As a result, code executing in the Non-Secure world cannot read or modify Secure-world code or data, even in the presence of a fully compromised Non-Secure software stack.

    \item \textbf{Restricted Secure-World Entry.}
    TrustZone-M restricts access to Secure-world functionality to a set of explicitly defined entry points located in the \emph{Non-Secure Callable} (NSC) region.
    This region is part of Secure memory and is protected from modification by Non-Secure software.
    Together with TrustZone-M's secure state transition mechanism (enforced in hardware), this restriction ensures that Secure-world functions are invoked only in their entirety.
\end{itemize}

Recent studies of TrustZone-M based IoT systems highlight that misconfiguration,
improper partitioning of peripherals, and weak key storage practices can undermine the intended isolation guarantees, leading to practical exploits even when the hardware
security extensions are present~\cite{tzm-sec-iot,silabs-secure-key-storage}.
These findings motivate our use of TrustZone-M as a tightly scoped root of trust: \rtsmarm confines security-critical state (e.g., attestation keys and permutation
metadata) to the Secure world, while treating the Non-Secure RTOS and application stack as fully adversarial.

\section{Realization of SMARM on RTOS-based Systems}
\label{sec:porting}

This section describes our efforts to realize SMARM in RTOS systems.
We begin by discussing the non-trivial challenges that preclude a simple extension of existing SMARM prototype (HYDRA-based one) to RTOS, thereby motivating the need for a new direction in this work.

\subsection{Challenges of Adapting SMARM to RTOS}\label{subsec:naive}

A naive approach is to adapt HYDRA-based prototype of SMARM to operate atop an RTOS.
In principle, this would correspond to replacing the seL4 microkernel with an RTOS kernel while leaving user-level tasks (including the measurement task) unchanged.
However, such an approach raises several security issues.

Recall the HYDRA-based SMARM prototype in Figure~\ref{fig:hydra_impl}.
First, to satisfy (S1)-(S3), it relies on seL4's isolation guarantees to protect the measurement task (that implements the SMARM logic) from other untrusted tasks.
These isolation guarantees are absent in commonly deployed RTOSes.
For example, vanilla FreeRTOS~\cite{freertos}, the most commonly used RTOS in commercial real-time embedded systems, operates under bare-metal settings in which all tasks execute at the same privilege level.
As a result, other (potentially compromised) tasks can freely access all system resources, including $k$ and $\sigma$ used by the measurement task, hence easily breaking (S2) and (S3).
Similarly, the default configuration of NuttX~\cite{nuttx} provides no inter-task isolation.

Even when considering RTOSes that offer some form of memory isolation, e.g., the MPU version of FreeRTOS~\cite{freertos-mpu} or Zephyr OS~\cite{zephyr}, these protections are not backed by formal verification, unlike seL4.
This means their security guarantees only hold as long as the kernel itself remains uncompromised.
In practice, these RTOSes have been shown to contain many kernel-level vulnerabilities, e.g., CVE-2018-16525 that enables remote code execution in FreeRTOS or CVE-2024-6137 that leads to malicious writes to Zephyr's kernel memory~\cite{cve_freertos_rce,cve_zephyr_memory}.
An adversary can exploit such vulnerabilities to break isolation, rendering SMARM insecure on these platforms.

Therefore, simply replacing seL4 with an RTOS of choice is insufficient to make SMARM \emph{secure}.
Instead, our design leverages ARM TrustZone-M to provide hardware-enforced isolation mechanisms available on off-the-shelf embedded platforms.
Next, we describe the system model considered in this work, followed by the corresponding threat model.

\subsection{System \& Threat Model}

\begin{figure}[!htp]
    \centering
    \includegraphics[width=1.0\linewidth]{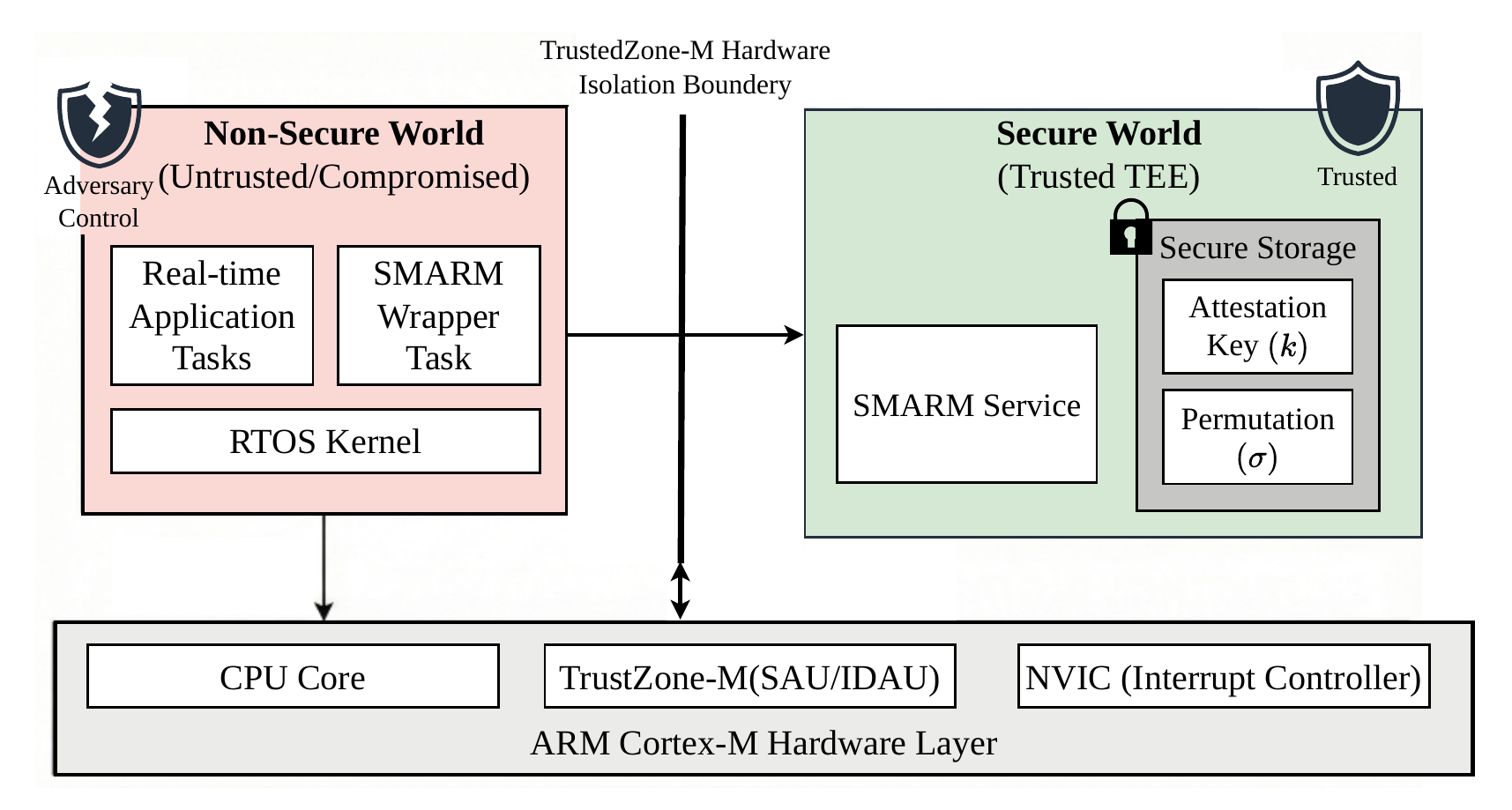}
    \caption{Our proposal of SMARM on RTOS using TrustZone-M}
    \label{fig:system_threat_model}
\end{figure}

We consider a prover implemented on an ARM Cortex-M device equipped with TrustZone-M, which divides execution into two runtime environments: the Secure world (trusted TEE) and Non-Secure worlds (untrusted).
The Non-Secure world hosts untrusted software including the RTOS kernel and real-time application tasks.
In contrast, SMARM's security-critical logic is assumed to be trusted, resides in the Secure world and is accessible from the Non-Secure world only via a NSC call.

We adopt the same threat model as SMARM: an adversary can exploit software vulnerabilities to gain full control over all untrusted software.
Under our system model, this corresponds to the adversary fully compromising the entire Non-Secure world.
Since the Secure world is assumed to be trusted, the adversary cannot tamper with the SMARM logic residing in the Secure world.
In addition, we consider KFV malware that can interrupt Secure-world execution after each block measurement (via NVIC control; see Section~\ref{subsec:tzm}), with the aim of relocating malicious code across memory blocks to evade detection.
Following SMARM, we consider hardware/physical attacks, side-channel attacks and denial-of-service (DoS) attacks to be out of scope.

\subsection{SMARM on RTOS using TrustZone-M}

We propose a new realization of SMARM that coexists with an RTOS and real-time application tasks without requiring custom hardware.
Unlike the naive approach discussed in Section~\ref{subsec:naive}, this design leverages the hardware-enforced isolation provided by TrustZone-M, retaining SMARM's security even if an adversary fully compromises all untrusted software in the Non-Secure world.

\subsubsection{Components}

In addition to the unmodified RTOS kernel and real-time application tasks, we decompose SMARM into two components, as illustrated in Figure~\ref{fig:system_threat_model}.

\textbf{Wrapper Task.}
The Wrapper Task is implemented as a regular RTOS task executing in the Non-Secure world.
It serves as an interface between the verifier (or other tasks responsible for prover-verifier communication) and the Secure-world SMARM logic.
Upon receiving an attestation request, the wrapper task invokes the Secure Service, obtains the attestation evidence, and forwards it to the verifier.
To allow preemption by real-time tasks, we configure the wrapper task to run at the lowest priority.
    
\textbf{Secure Service.}
The Secure Service resides entirely in the Secure world and is exposed to the Non-Secure world via an NSC entry point.
When invoked with an attestation request, it executes the SMARM logic as described in Section~\ref{subsec:smarm} and detailed in Algorithm~\ref{alg:smarm_baseline}, then returning the attestation evidence to the caller.
To enforce atomic block measurements, the service disables interrupts via the \texttt{CPSID} instruction immediately before each block measurement and re-enables them upon completion.
As a result, (S1) is guaranteed in this design by preventing partial execution of the Secure Service and disallowing interruptions during block measurements.
    
To satisfy (S2), we store the attestation key~$k$ persistently in Secure storage that is accessible only from the Secure world.
At runtime, all $k$-related materials reside in Secure RAM, which is inaccessible to the Non-Secure world.
Similarly, to satisfy (S3), the permutation~$\sigma$, once generated, is not exposed to Non-Secure RAM and therefore cannot be inferred by KFV malware.
KFV malware cannot also perform replay attacks since $\sigma$ is based on a one-time challenge $chal$ and a secret $k$.

\begin{algorithm}[t]
\caption{Pseudocode of SMARM Service}
\label{alg:smarm_baseline}
\begin{algorithmic}[1]
\small
\Notation Memory $M$, block size $B$, number of blocks $n$, key $k$, challenge $chal$
\State $\text{seed} \gets \mathsf{Hash}(chal, k)$
\State $\sigma = \mathsf{FisherYates}(\text{seed}, [0,\ldots,n-1])$
\State \Call{HMAC.Init}{$k$}
\For{$i \gets 0$ to $n-1$}
    \State \Call{DisableInterrupts}{}
    \State \Call{HMAC.Update}{hmac, $M_{\sigma[i]}$, $B$}
    \State \Call{EnableInterrupts}{}
\EndFor
\State att-evidence $\gets$ \Call{HMAC.Finalize}{hmac}
\Output att-evidence
\end{algorithmic}
\end{algorithm}

\subsubsection{Security}
\emph{Since the SMARM Service satisfies all three security requirements (S1)–(S3), our new realization of SMARM retains the same security guarantees as the original SMARM design~\cite{smarm}.}
The SMARM Wrapper Task, in contrast, is not part of the trusted computing base (TCB) and serves only a functional role: it enables attestation requests from the untrusted Non-Secure world to invoke the SMARM Service.
Accordingly, this wrapper task can be treated as untrusted and executed entirely in the Non-Secure world.
Indeed, an adversary may compromise this wrapper; 
however, such a compromise does not enable the adversary to forge valid attestation evidences or to gain more information than what is already achievable by KFV malware.
At most, a compromised wrapper task can refuse to call the SMARM Service, resulting in DoS attacks that are outside the scope of this work.
A more formal proof (based on cryptographic reduction) of our realization of SMARM on RTOS and TrustZone-M can be found in Appendix~\ref{sec:formal-security}.



\subsubsection{Workflow}

\begin{figure}[!htp]
\centering
\includegraphics[width=\linewidth]{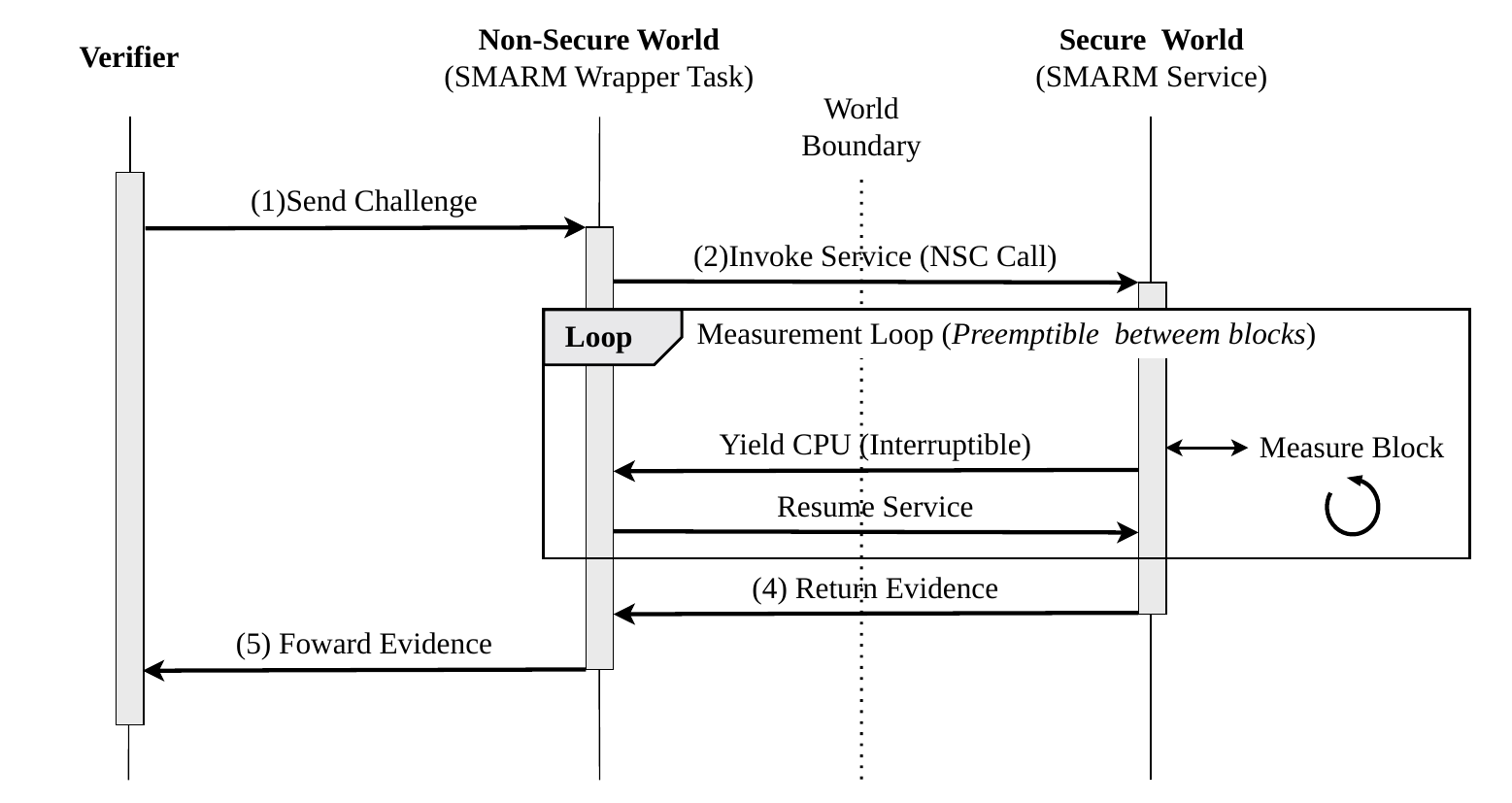}
\caption{The end-to-end workflow of our new SMARM proposal using TrustZone-M}
\label{fig:smarm_workflow}
\end{figure}

To complete this section, we describe the end-to-end interaction between the verifier and the SMARM components on the prover, as illustrated in Figure~\ref{fig:smarm_workflow}:

\begin{enumerate}
    \item The verifier sends a fresh challenge to the prover. The challenge is received by the SMARM Wrapper Task executing in the Non-Secure world; optionally, other tasks may receive the challenge and forward it to the wrapper.
    
    \item Upon receiving the challenge, the SMARM Wrapper Task invokes the SMARM Service through an NSC entry point.
    
    \item The SMARM Service executes the SMARM logic inside the Secure world, following Algorithm~\ref{alg:smarm_baseline}.
    Notably, after each block measurement completes, execution may return to the Non-Secure world, where the RTOS can preempt the SMARM Wrapper Task (which runs at the lowest priority) to schedule real-time tasks.

    \item After completing the attestation, the SMARM Service returns the attestation evidence to the SMARM Wrapper Task.
    
    \item Finally, the SMARM Wrapper Task forwards the evidence to the verifier. If the verifier receives no or invalid evidence, it treats the prover as compromised.
\end{enumerate}

\section{Evaluation of SMARM on RTOS-based Systems}
\label{sec:evaluation}

This section evaluates our TrustZone-M-based SMARM realization on RTOS platforms, focusing on implementation details, non-interruptibility behavior, and real-time availability under different workloads.

\subsection{Implementation}

We implement the new realization of SMARM on an STM32 Nucleo-144 development board~\cite{stm32l552_nucleo144}, which is equipped with an ARM Cortex-M33 MCU running at 110MHz, 512KB of flash storage, and 256KB of SRAM;
it also supports the TrustZone-M security extension.
In our configuration, 256KB of flash and 128KB of SRAM are allocated to the Secure world, with the remaining memory assigned to the Non-Secure world.

On the software side, the Non-Secure world runs an unmodified RTOS, where our prototype supports both FreeRTOS and Zephyr, compiled with size optimizations (-Os).
We implement the SMARM Wrapper Task as a regular RTOS task in the Non-Secure world, consisting of \textit{115} lines of C code.

The Secure-world software implements the SMARM Service as detailed 
in Algorithm~\ref{alg:smarm_baseline}, where we use SHA-256 as the 
hash function for memory block measurement and HMAC-SHA256 to compute 
the attestation evidence. The Secure-world implementation consists of 
\textit{153} lines of C code and contributes \textit{33{,}968}~bytes 
to TCB.

\subsection{Non-interruptibility Period}

\begin{figure}[!htp]
    \centering
    \includegraphics[width=0.9\linewidth]{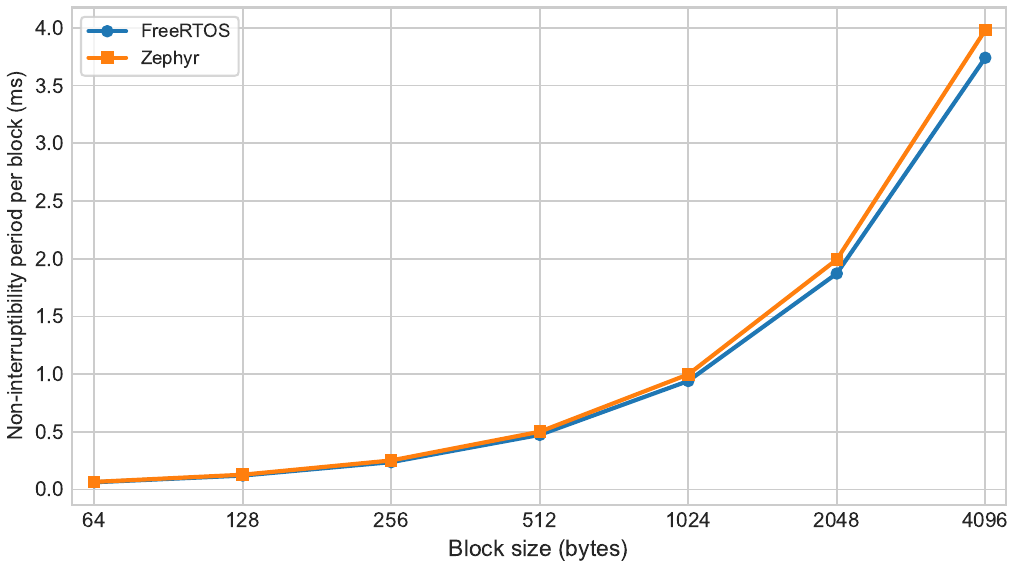}
   \caption{Non-interruptibility period $t_{\text{disabled}}$ of our SMARM proposal on the STM32 Nucleo-144 platform. Each point is averaged over five runs and standard deviations are below 1\%.}
    \label{fig:tdisabled_nucleo}
\end{figure}

We first evaluate SMARM on our STM32 Nucleo-144 prototype, focusing on the non-interruptibility period $t_{\text{disabled}}$, i.e., the time spent in an atomic block measurement.
Since $t_{\text{disabled}}$ depends on the block size $B$, we consider
$B \in \{64, 128, 256, 512, 1024, 2048, 4096\}$~bytes and run SMARM using both FreeRTOS and Zephyr as the underlying RTOS.

Figure~\ref{fig:tdisabled_nucleo} reports the measured $t_{\text{disabled}}$ for each block size.
As expected, $t_{\text{disabled}}$ grows approximately linearly with $B$, confirming that larger blocks lead to longer non-interruptibility periods on our target device.
Together with the storage analysis in Figure~\ref{fig:smarm_tdisabled}, the results show that our realization preserves SMARM's original performance trade-offs~\cite{smarm}.
However, these trends do not directly reveal the impact on real-time availability, which we examine in the next subsection.

\subsection{Real-time Availability: Setup \& Metric}

\paragraph{\textbf{Experimental Setup}}
To evaluate the impact of SMARM on real-time availability in RTOS-based systems, we consider a simplified yet representative scenario of real-time settings where SMARM runs alongside a single real-time task $T_{RT}$.

In particular, we model three classes of practical real-time workloads by varying $T_{RT}$ frequency:
\begin{enumerate}
    \item \textbf{High} real-time workload where $T_{RT}$ runs at 1000Hz
    \item \textbf{Medium} real-time workload -- 100Hz
    \item \textbf{Low} real-time workload -- 10Hz
\end{enumerate}

Although our setup includes only one real-time task, it is intentionally conservative.
Our purpose is to identify the conditions under which SMARM and real-time settings fail to co-exist.
In other words, if this co-existence does not hold in this simplified setting, it is unlikely to do so in more complex deployments with multiple real-time tasks. We measure $f_{\text{observed}}$ using the on-chip hardware cycle counter and convert cycles to time using the 110\,MHz CPU clock.


\paragraph{\textbf{Metric}}
To quantify the impact of SMARM on real-time availability, we 
introduce the \emph{Frequency Accuracy Ratio} (FAR), defined as
\begin{equation}
    FAR = \frac{f_{\text{observed}}}{f_{\text{baseline}}}.
\end{equation}
Here, $f_{\text{observed}}$ denotes the execution frequency of 
the real-time task $T_{RT}$ when SMARM is active (including both 
the SMARM Wrapper Task and SMARM Service).
In contrast, $f_{\text{baseline}}$ denotes the baseline frequency 
of $T_{RT}$ in the absence of SMARM-induced interference, i.e., 
SMARM without disabled interrupts (Lines 5 and 7 in 
Algorithm~\ref{alg:smarm_baseline}).
Intuitively, FAR measures how well the real-time task maintains 
its intended execution rate in the presence of SMARM.
$FAR = 1.0$ means that SMARM introduces no interference due to 
non-interruptible operations.
As these non-interruptible operations increasingly delay task 
execution, the observed frequency of $T_{RT}$ decreases, leading to a lower FAR.~\\


\noindent\textbf{Remark.}
FAR does not replace schedulability analysis.
It measures the frequency-level interference caused by SMARM's non-interruptible block measurements, assuming $T_{RT}$ is already schedulable in the baseline SMARM configuration where interrupts are not disabled.
If this assumption does not hold, the problem is an orthogonal scheduling issue independent of SMARM.

\begin{figure}[!ht]
    \centering
    \begin{subfigure}[t]{0.9\linewidth}
        \centering
        \includegraphics[width=\linewidth]{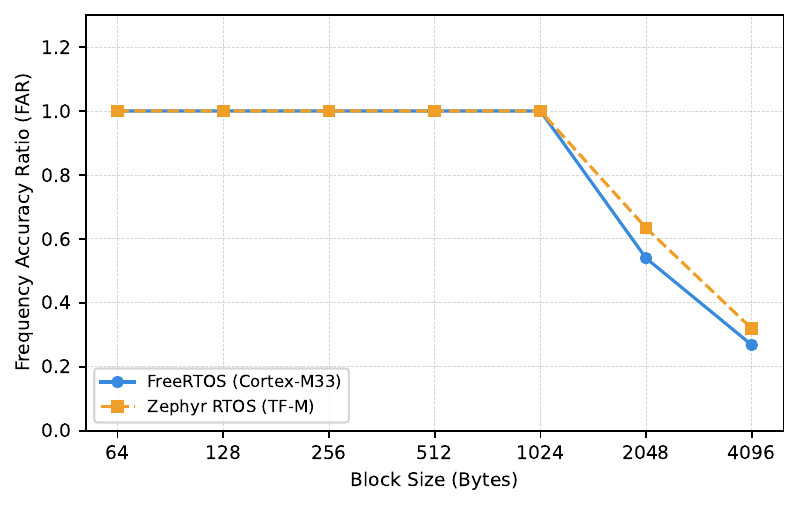}
        \caption{High real-time workload}
        \label{fig:baseline_results_1000hz}
    \end{subfigure}
    \vspace{0.6em}
    \begin{subfigure}[t]{0.9\linewidth}
        \centering
        \includegraphics[width=\linewidth]{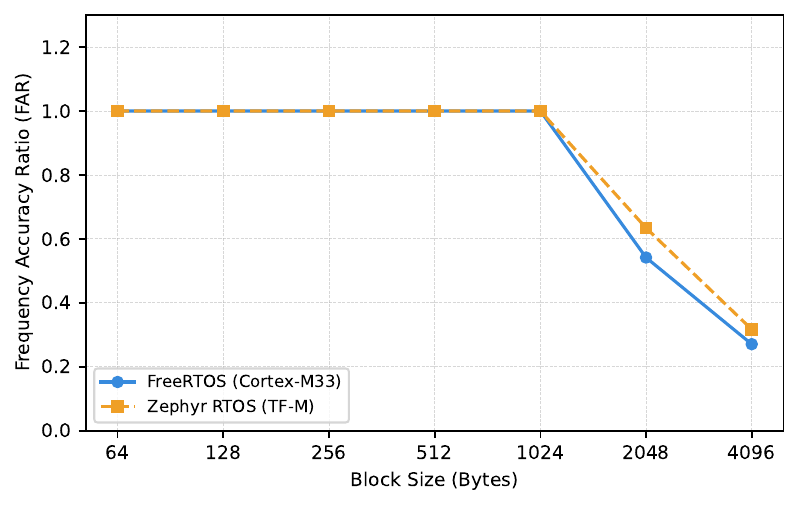}
        \caption{Medium real-time workload}
        \label{fig:baseline_results_100hz}
    \end{subfigure}
    \vspace{0.6em}
    \begin{subfigure}[t]{0.9\linewidth}
        \centering
        \includegraphics[width=\linewidth]{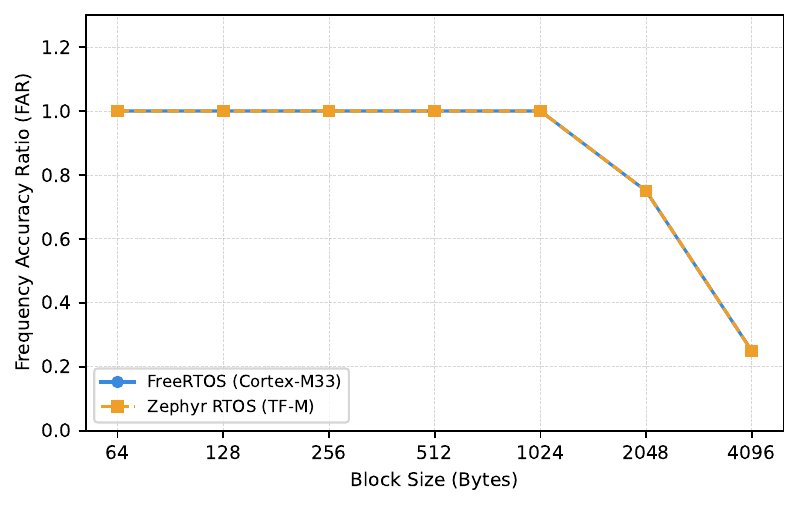}
        \caption{Low real-time workload}
        \label{fig:baseline_results_10hz}
    \end{subfigure}
    \caption{Real-time availability of SMARM under different real-time workloads}
    \label{fig:setup_and_baseline_far}
\end{figure}

\subsection{Real-time Availability: Results}

Figure~\ref{fig:setup_and_baseline_far} reports FAR across different block sizes, real-time workloads, and RTOS choices.
Across all configurations, we observe a consistent trend: FAR remains close to $1.0$ for $B \leq 1024$~bytes, but degrades noticeably as the block size increases.

At $B = 2048$~bytes, FAR decreases to approximately $0.55$ on FreeRTOS and $0.65$ on Zephyr under medium/high real-time workloads, while remaining around $0.75$ under low workloads.
At $B = 4096$~bytes, FAR drops further to $0.27$ on FreeRTOS and $0.32$ on Zephyr across all workload settings.
These results indicate that SMARM cannot reliably preserve real-time availability at these larger block sizes.

Both RTOS-s exhibit nearly identical FAR trends.
This suggests that the dominant source of interference is SMARM’s atomic block measurement rather than RTOS-specific scheduling behavior.
In particular, the block size $B$, which directly determines the non-interruptibility period, emerges as the primary factor affecting real-time availability.

However, these results also reveal a key limitation of SMARM.
To maintain $FAR \approx 1.0$ across all evaluated settings, we must use relatively small block sizes ($B \leq 1024$~bytes).
Under SMARM, smaller blocks imply a larger number of memory partitions and therefore increased secure storage for the permutation $\sigma$, ranging from $576$~bytes at $B = 1024$~bytes to approximately $13$~KB at $B = 64$~bytes.

Such storage demands can become problematic in Secure-world deployments that must coexist with other services, such as confidential DNN inference~\cite{tslices}, control-flow integrity monitoring~\cite{tzmcfi}, or secure key storage~\cite{silabs-secure-key-storage}.
In these scenarios, dedicating between several hundred bytes and tens of kilobytes of Secure RAM solely for storing $\sigma$ may be impractical.

\section{\rtsmarm: Reducing Secure Storage in SMARM}
\label{sec:resmarm-design}

In this section, we propose \rtsmarm, enhanced versions of SMARM that reduce secure storage requirements while preserving its real-time behavior and security guarantees.

\rtsmarm retains the same high-level structure as SMARM, consisting of the Wrapper Task and the Secure Service.
The Wrapper Task remains unchanged, while we redesign the Secure Service to reduce the secure storage overhead.

The key idea of \rtsmarm is to trade secure storage for additional computation, thereby enabling the use of smaller block sizes under tight secure-memory budgets.
We present two variants of \rtsmarm.
The first variant, \rtsmarmone (Section~\ref{sec:rtsmarm-bitarray}), reduces the secure storage requirement from $O(n \log n)$ to $O(n)$.
The second variant, \rtsmarmtwo (Section~\ref{sec:rtsmarm-fpe}), eliminates the need for secure storage altogether, at the cost of significantly longer overall runtime.

\subsection{\rtsmarmone\ -- Reducing Secure Storage to $O(n)$}
\label{sec:rtsmarm-bitarray}

\begin{algorithm}[!t]
\caption{\rtsmarmone (secure world)}
\label{alg:rtsmarm-bitarray}
\begin{algorithmic}[1]
\small
\Notation Memory $M$, block size $B$, number of blocks $n$, key $k$, challenge $chal$
\State \colorbox{codehl}{\textbf{for all} $i \in \{0,\dots,n-1\}$: $\text{used}[i] \gets 0$} 
       \Comment{new: bitmap initialisation}
\State \colorbox{codehl}{$\text{seed} \gets \mathsf{Hash}(chal, k)$}
\State \colorbox{codehl}{$\mathsf{PRNG.Init}(\text{seed})$}
\State \Call{HMAC.Init}{$k$}
\For{$rm \gets 0$ to $n-1$}
    \State \colorbox{codehl}{$\text{rem} \gets n - rm$}
    \State \colorbox{codehl}{$r \gets \mathsf{PRNG.RandInt}(\text{1}, \text{rem})$}
    \State \colorbox{codehl}{$idx \gets \mathsf{FindKthZero}(\text{used}, r)$}
    \State \colorbox{codehl}{$\text{used}[idx] \gets 1$}
    \State \Call{DisableInterrupts}{}
    \State \Call{HMAC.Update}{hmac, $M_{idx}$, $B$}
    \State \Call{EnableInterrupts}{}
\EndFor
\State att\_evidence $\gets$ \Call{HMAC.Finalize}{hmac}
\Output att\_evidence
\end{algorithmic}
\end{algorithm}

Algorithm~\ref{alg:rtsmarm-bitarray} describes the operation of the Secure Service in \rtsmarmone (note that the Wrapper Task remains unchanged).
Instead of storing the full permutation $\sigma$ to determine the measurement order, \rtsmarmone uses a bitmap \texttt{used} of size $n$ to track which memory blocks have already been measured, i.e., \texttt{used}[i] contains a bit indicating whether block $M[i]$ has been measured by \rtsmarmone.

\rtsmarmone first initializes \texttt{used} to all zeros (Line~1), indicating that no blocks have been measured.
It then initializes a PRNG using the challenge $chal$ and key $k$ (Lines~2--3).

For each measurement iteration, \rtsmarmone determines the next block to measure as follows.
It first computes the number of remaining (unmeasured) blocks, denoted as $rem$ (Line~6) and then generates a random integer $r \in [1, rem]$ (Line~7).
Since $r$ does not directly correspond to a valid block index (some blocks may have already been measured), \rtsmarmone scans the \texttt{used} bitmap to locate the $r$-th zero entry and maps it to a block index $idx$ (Line~8).

The selected block is then marked as measured by setting $\texttt{used}[idx] \leftarrow 1$ (Line~9).
Subsequently, $M_{idx}$ is measured following the same way as in SMARM.
This process repeats until all blocks have been covered, after which the attestation evidence is computed and becomes the output of the Secure Service (Line~14).

\subsubsection{Secure Storage vs.\ Runtime Trade-off.}
Compared to SMARM, \rtsmarmone replaces the secure storage of the permutation $\sigma$ with a bitmap \texttt{used}.
Since \texttt{used} contains $n$ entries, it requires $n$ bits of secure storage.
This reduces the storage requirement from $n \cdot \lceil \log_2 n \rceil$ bits in SMARM to $n$ bits in \rtsmarmone.

This reduction in storage comes at the cost of potential additional runtime overhead.
In SMARM, the permutation $\sigma$ is generated once (e.g., via \texttt{FisherYates} in Algorithm~\ref{alg:smarm_baseline}) before the measurement loop, after which each block index can be retrieved in constant time.
In contrast, \rtsmarmone determines the next block indices on-the-fly by scanning the \texttt{used} bitmap to locate the $r$-th unmeasured block.
This requires traversing up to $O(n)$ entries per iteration.

Consequently, while \rtsmarmone reduce secure storage to $O(n)$, it introduces additional per-block computation, which can increase the overall attestation runtime, especially for large $n$.

\subsubsection{Soundness.}\label{subsubsec:v1-sound}
We first establish the soundness of \rtsmarmone, i.e., showing that \rtsmarmone is functionally equivalent to SMARM.
In this context, soundness requires that (1) all memory blocks are measured exactly once, and (2) the order of block selection is randomized.

For (1), in each iteration, the algorithm selects an index $idx$ such that \texttt{used}[idx] = 0, and then sets it to $1$.
Since each iteration flips exactly one bit from $0$ to $1$, and there are $n$ such bits, after $n$ iterations all entries in \texttt{used} are set to $1$.
This means that every block is selected exactly once, and no block is measured more than once.

For (2), randomness follows from the use of PRNG seeded with the secret (random) key $k$.
At each iteration, \rtsmarmone samples a uniform random integer $r \in [1, rem]$.
The mapping from $r$ to the $r$-th zero entry in \texttt{used} establishes a one-to-one correspondence between $r$ and the set of remaining blocks.
Therefore, each unmeasured block is selected with equal probability at each step.
As a result, the overall measurement order is a uniformly random permutation over all $n$ blocks.

\subsubsection{Security}
Here, we provide informal arguments that \rtsmarmone retains the same security as SMARM by fulfilling the same security guarantees (S1)-(S3) defined in Section~\ref{sec:background}; we refer to Appendix~\ref{sec:proof-prng} for a formal (reduction-based) security proof of \rtsmarmone.

\begin{itemize}
  \item[(S1)] \emph{Controlled execution.}
  Recall in Section~\ref{sec:porting} that SMARM guarantees (S1) by (i) relying on TrustZone-M to restrict entry into the Secure world to well-defined entry points (Section~\ref{subsec:tzm}), and (ii) disabling interrupts during each block measurement to ensure atomic execution.
  \rtsmarmone preserves both design choices by implementing its logic as a Secure-world service and maintaining interrupt disablement during block measurement.
  Therefore, (S1) is directly inherited in \rtsmarmone.

  \item[(S2)] \emph{Key secrecy.}
  \rtsmarmone does not modify the generation, storage, or usage of the attestation key $k$.
  As in SMARM, $k$ is stored and accessed exclusively within the Secure world.
  Hence, \rtsmarmone preserves (S2).

  \item[(S3)] \emph{Private measurement order.}
  In \rtsmarmone, the measurement order is implicitly determined by the \texttt{used} bitmap, which resides in Secure RAM inaccessible to the Non-Secure world where KFV malware may reside.
\end{itemize}

\subsection{\rtsmarmtwo\ -- Eliminating Secure Storage}
\label{sec:rtsmarm-fpe}

We observe that SMARM instantiates PRP using the Fisher--Yates shuffle~\cite{fisherYates}, which produces a random permutation over $n$ elements.
While Fisher--Yates is time-efficient (i.e., in $O(n)$), it requires materializing the entire permutation, resulting in $O(n \log n)$ bits of secure storage in SMARM.

In \rtsmarmtwo, we eliminate this storage overhead by instead instantiating the PRP using format-preserving encryption (FPE)~\cite{bellare2009format,ff1}.
FPE is a symmetric encryption primitive that preserves the domain of its inputs, i.e., for a given encryption key from the set $\mathcal{K}$ and a given finite domain $\mathcal{D}$, it defines an encryption as:
\[
\mathsf{FPE.Enc}: \mathcal{K} \times \mathcal{D} \rightarrow \mathcal{D}
\]
By construction, $\mathsf{FPE.Enc}$ realizes a bijection over $\mathcal{D}$, and thus directly serves as a PRP.

\begin{algorithm}[!t]
\caption{\rtsmarmtwo (secure world)}
\label{alg:rtsmarm-fpe}
\begin{algorithmic}[1]
\small
\Notation Memory $M$, block size $B$, number of blocks $n$, key $k$, challenge $chal$
\State \colorbox{codehl}{$(\text{ffx\_key}, \text{tweak}) \gets \mathsf{KDF}(k, chal)$}
\State \colorbox{codehl}{$\mathsf{FFX.Init}(\text{ffx\_key, tweak})$}
\State \Call{HMAC.Init}{$k$}
\For{$i \gets 0$ to $n-1$}
    \State \colorbox{codehl}{$idx \gets \mathsf{FFX.Enc}(i)$}
    \State \Call{DisableInterrupts}{}
    \State \Call{HMAC.Update}{hmac, $M_{idx}$, $B$}
    \State \Call{EnableInterrupts}{}
\EndFor
\State att\_evidence $\gets$ \Call{HMAC.Finalize}{hmac}
\Output att\_evidence
\end{algorithmic}
\end{algorithm}

In \rtsmarmtwo, we instantiate FPE using the FFX construction~\cite{bellare2010ffx}.
FFX implements FPE using a Feistel network~\cite{hoang2010generalized}, where each round function can be realized as a block cipher.
To keep the design lightweight, we use Speck~\cite{speck}, a block cipher optimized for efficient software execution on resource-limited platforms, which aligns with our target deployment setting.

Algorithm~\ref{alg:rtsmarm-fpe} describes how \rtsmarmtwo integrates FFX into the block measurement process.
First, \rtsmarmtwo derives the FFX key and tweak from the attestation key $k$ and the challenge $chal$ (Line~1), and initializes the FFX instance accordingly (Line~2).

During the $i$-th iteration, \rtsmarmtwo determines the next block to measure by using the instantiated FFX to encrypt $i$ (Line~5).
The resulting $idx$ is then used to retrieve and measure the corresponding memory block, following the same procedure as in SMARM.
After all blocks have been measured, \rtsmarmtwo computes and returns the attestation evidence in the same manner as SMARM.

\subsubsection{Secure Storage vs.\ Runtime Trade-off.}
\rtsmarmtwo eliminates the need for any per-block permutation state in secure memory.
This leaves \rtsmarmtwo secure storage with only the FPE encryption key and the HMAC state, both of which are $O(1)$ in size and independent of~$n$.

This elimination of permutation storage comes at the cost 
of additional per-block (non-atomic) computation.
In SMARM, each block index is retrieved in constant time 
from the pre-computed $\sigma$.
In \rtsmarmone, each iteration requires an $O(n)$ scan of 
the \texttt{used} bitmap.
In \rtsmarmtwo, each iteration invokes 
$\mathsf{FFX.Enc}$, which performs one FPE encryption using Speck as the underlying round function. 
While this operation does not depend on $n$, i.e., $O(1)$ runtime complexity, this operation is more complex than the scan operation in \rtsmarmone, potentially making its runtime longer for a smaller $n$.

Consequently, \rtsmarmtwo trades all permutation storage 
for a constant per-block cryptographic cost, making it 
the most storage-efficient variant at the expense of higher per-block computation compared to SMARM.

\subsubsection{Soundness.}
Similar to Section~\ref{subsubsec:v1-sound}, we argue the soundness by showing that in \rtsmarmtwo, (1) all memory blocks are measured exactly once, and (2) the order of block selection is randomized.

\rtsmarmtwo iterates over $i \in \{0, \ldots, n-1\}$ 
and maps each $i$ to a block indices 
$idx = \mathsf{FFX.Enc}(i)$.
Since $\mathsf{FFX.Enc}$ implements a bijection over the domain 
$\{0, \ldots, n-1\}$, each value of $idx$ is distinct for 
distinct values of $i$.
Therefore, every block is measured exactly once and no block 
is measured more than once, fulfilling (1).

(2) directly follows from the fact that FFX instantiates pseudo-random permutation and the order of block selection in \rtsmarmtwo is simply\[
  \{\mathsf{FFX.Enc}(0), \mathsf{FFX.Enc}(1), \ldots,
    \mathsf{FFX.Enc}(n-1)\},
\] making it a uniformly random permutation over all $n$ blocks.

\subsubsection{Security}
We informally argue how \rtsmarmtwo satisfies (S1)-(S3) (in Section~\ref{sec:background}) while providing a formal reduction proof in Appendix~\ref{sec:proof-fpe}.

\begin{itemize}
  \item[(S1)] \emph{Controlled execution.}
  \rtsmarmtwo preserves both design choices of SMARM for 
  enforcing (S1): it implements its logic as a Secure-world 
  service accessible only via TrustZone-M NSC entry points 
  (Section~\ref{subsec:tzm}), and it disables interrupts 
  during each block measurement via the \texttt{CPSID} 
  instruction to ensure atomic execution.
  Therefore, (S1) is directly inherited in \rtsmarmtwo.

  \item[(S2)] \emph{Key secrecy.}
  \rtsmarmtwo does not modify the generation, storage, or 
  usage of the attestation key $k$.
  As in SMARM, $k$ is stored and accessed exclusively within 
  the Secure world and is used only to derive the FFX key and tweak, which also do not leak from the Secure world.
  Hence, \rtsmarmtwo preserves (S2).

  \item[(S3)] \emph{Private measurement order.}
  In \rtsmarmtwo, the measurement order is determined by 
  $\mathsf{FFX.Enc}$, whose output depends on the 
  secret key $k$ and the one-time challenge $chal$.
  Since $k$ resides exclusively in Secure RAM inaccessible to the Non-Secure world, KFV malware cannot compute or predict the permutation.
  %

\end{itemize}

\section{Evaluation of \rtsmarm on RTOS Systems}
\label{sec:rtsmarm-evaluation}

In this section, we evaluate the two \rtsmarm variants against the
baseline SMARM along three dimensions: secure storage 
size, overall attestation runtime, and energy consumption. We reuse the experimental setup and FAR metric
introduced in Section~\ref{sec:evaluation}. 
As established in Section~\ref{sec:evaluation} that FreeRTOS and Zephyr exhibit similar results in SMARM, our experiments here therefore focus only on FreeRTOS.

\subsection{Secure Storage}
\label{sec:storage-runtime-tradeoff}

\begin{figure}[!ht]
    \centering
    \includegraphics[width=0.95\linewidth]{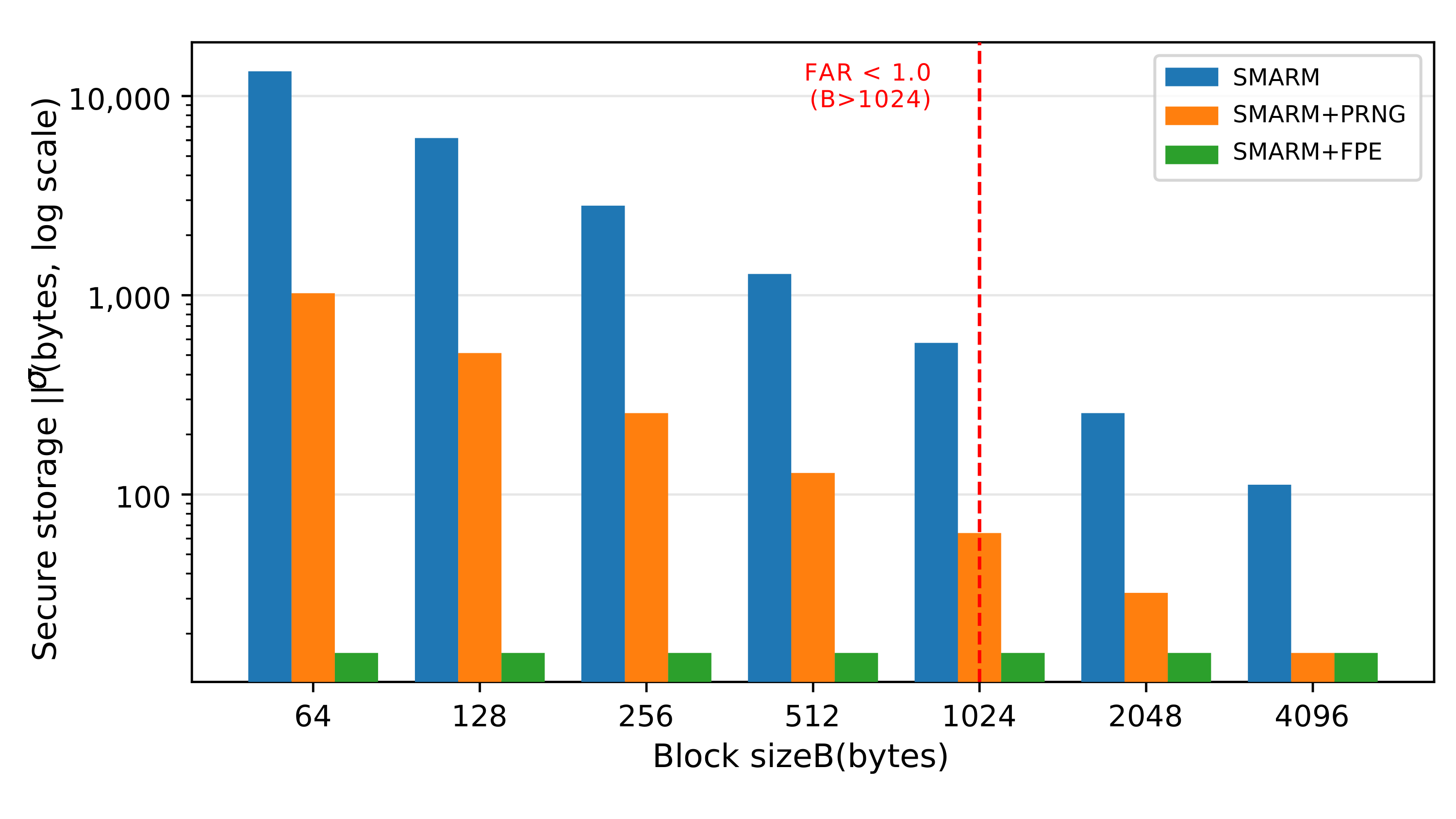}
    \caption{Storage requirement for each scheme with varying block sizes. The dashed red line marks $B = 1024$~bytes, beyond which 
    FAR drops below~1.0.}
    \label{fig:storage_tradeoff}
\end{figure}

\begin{table}[!ht]
  \centering
  \small
  \setlength{\tabcolsep}{6pt}
  \begin{tabular}{rrrr}
    \toprule
    & \multicolumn{3}{c}{\textbf{Runtime (ms)}} \\
    \cmidrule(lr){2-4}
    $B$ \textbf{(bytes)} 
    & \textbf{SMARM} 
    & \textbf{\rtsmarmone} 
    & \textbf{\rtsmarmtwo} \\
    \midrule
    64    & 644 & 9{,}310 & 4{,}229 \\
    128   & 562 & 2{,}730 & 2{,}325 \\
    256   & 521 & 1{,}066 & 1{,}423 \\
    512   & 500 &   630   &   945   \\
    \rowcolor{gray!15}
    1024  & 490 &   526   &   726   \\
    2048  & 484 &   494   &   611   \\
    4096  & 481 &   484   &   553   \\
    \bottomrule
  \end{tabular}
  \caption{Overall attestation runtime of SMARM baseline and the two \rtsmarm\ variants. The shaded row ($B = 1024$~bytes) marks $B=1024$~bytes, beyond which FAR drops below 1.0.}
  \label{tab:runtime_overhead}
\end{table}

Figure~\ref{fig:storage_tradeoff} compares the secure storage requirements of SMARM and both \rtsmarm variants across different block sizes $B$.
Both \rtsmarm schemes substantially reduce secure storage overhead relative to SMARM across all evaluated configurations.

The magnitude of reduction varies across the two variants as $B$ changes.
For \rtsmarmone, the storage benefit is greater when $n$ is large (equivalently, when $B$ is small).
For example, \rtsmarmone achieves a $13\times$ reduction at $B = 64$~bytes ($13{,}312 \rightarrow 1{,}024$~bytes), while the reduction decreases to approximately $7\times$ at $B = 4096$~bytes ($112 \rightarrow 16$~bytes).
This trend arises because both \rtsmarmone and SMARM's storage requirements scale with $n$, albeit with different asymptotic costs.

In contrast, \rtsmarmtwo requires a fixed secure storage footprint of only $16$~bytes for FPE-related key material, independent of $B$/$n$.
As a result, it achieves increasingly larger reductions as $B$ decreases.
At $B = 64$~bytes, \rtsmarmtwo provides an $832\times$ reduction compared to SMARM, far exceeding the reduction achieved by \rtsmarmone.
As $B$ increases, however, the relative benefit becomes narrower, and both variants provide comparable storage savings at large block sizes such as $B = 4096$~bytes.

At the largest block size that still achieves 1.0 FAR ($B = 1024$~bytes), \rtsmarmone reduces the storage requirement from $576$~bytes to $64$~bytes, corresponding to a $9\times$ reduction.
\rtsmarmtwo further lowers the requirement to $16$~bytes, achieving approximately a $36\times$ reduction relative to SMARM.

\begin{figure}[!ht]
    \centering
    \includegraphics[width=0.95\linewidth]{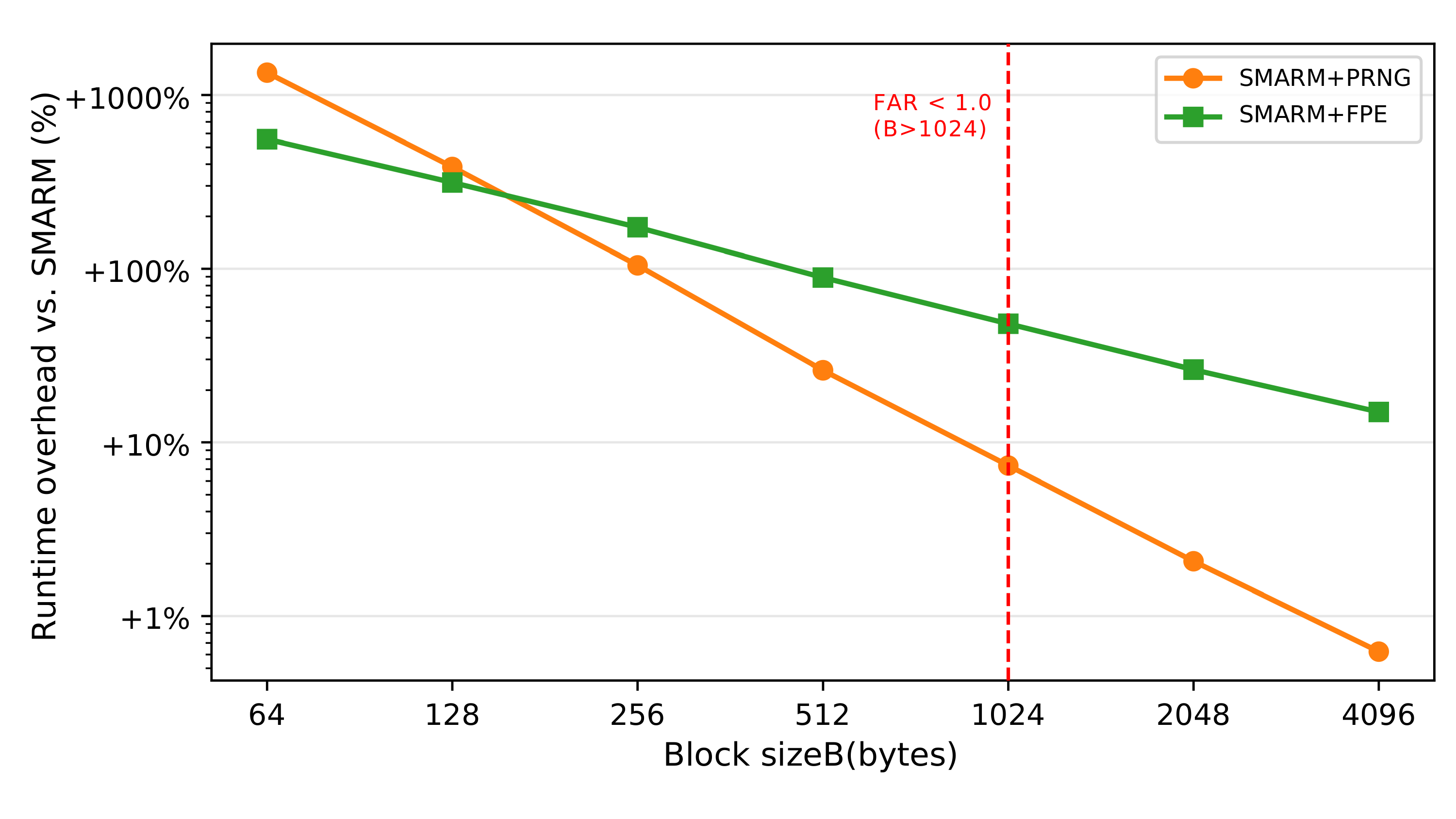}
    \caption{Overall runtime overhead of \rtsmarmone and \rtsmarmtwo relative to SMARM baseline across block sizes. The dashed red line marks $B = 1024$~bytes, beyond which 
    FAR drops below~1.0.}
    \label{fig:runtime_overhead}
\end{figure}

\subsection{Overall Attestation Runtime}
Figure~\ref{fig:runtime_overhead} presents the runtime overhead of \rtsmarmone and \rtsmarmtwo relative to SMARM (when executed without interruption), while the corresponding absolute runtimes are reported in Table~\ref{tab:runtime_overhead}.
The results show that both variants introduce additional runtime compared to SMARM, with the overhead becoming more pronounced as the block size $B$ decreases.

In most cases, \rtsmarmtwo tends to incur higher runtime overhead than \rtsmarmone due to the need to perform an FPE operation for each block selection.
However, for very small block sizes ($B = 64$ or $128$~bytes), \rtsmarmone becomes slower.
This is because its bitmap-based block selection requires scanning $O(n)$ entries per iteration, where $n$ grows inversely with $B$.
As a result, the cumulative scan cost becomes dominant for small blocks, leading to large overhead (e.g., $13.5\times$ of SMARM runtime at $B = 64$~bytes).
This makes \rtsmarmone less suitable for settings requiring very small block sizes.

At $B = 1024$~bytes, \rtsmarmone incurs only a moderate runtime increase of $7\%$ (526~ms vs.\ 490~ms).
For this configuration, the bitmap scan contributes little to the total execution time, which is largely dominated by hashing and HMAC-SHA256 computation over memory blocks.
In comparison, \rtsmarmtwo incurs around $48\%$ overhead (726~ms vs.\ 490~ms).

\subsection{Energy Consumption}
\label{sec:energy}

\begin{table}[!ht]
  \centering
  \small
  \setlength{\tabcolsep}{5pt}
  \resizebox{\linewidth}{!}{%
  \begin{tabular}{rrrr}
    \toprule
    & \multicolumn{3}{c}{\textbf{Board-level Energy (mJ)}} \\
    \cmidrule(lr){2-4}
    $B$ \textbf{(bytes)}
    & \textbf{SMARM}
    & \textbf{\rtsmarmone}
    & \textbf{\rtsmarmtwo} \\
    \midrule
    64    & 190.2 & 2{,}650 & 1{,}243 \\
    128   & 168.0 &    780  &    700  \\
    256   & 155.2 &    310  &    425  \\
    512   & 148.6 &    183  &    289  \\
    \rowcolor{gray!15}
    1024  & 145.2 &    153  &    221  \\
    2048  & 143.3 &    144  &    187  \\
    4096  & 141.0 &    141  &    170  \\
    \bottomrule
  \end{tabular}%
  }
  \caption{Board-level energy consumption per attestation
           round; the shaded row marks $B = 1024$\,bytes,
           the largest block size preserving
           FAR\,$\approx$\,1.0.}
  \label{tab:energy}
\end{table}

We measure board-level energy per attestation round using 
an FNIRSI FNB58 USB power meter in series between the host 
PC and the NUCLEO-L552ZE-Q board, sampling at 100~Hz. 
%
Results are reported in Table~\ref{tab:energy}.
SMARM shows only small energy variation across block sizes $B$, from 190~mJ at $B=64$ to 141~mJ at $B=4096$.
In contrast, both \rtsmarm variants show stronger dependence on~$B$: 
at $B = 64$, \rtsmarmone consumes 2{,}650~mJ 
(13.9$\times$ SMARM) due to its $O(n)$ bitmap scan, 
while \rtsmarmtwo consumes 1{,}243~mJ (6.5$\times$) 
from per-block \textsf{FFX.Enc} calls.
At the operating point $B = 1024$~bytes, overhead 
narrows to $5\%$ for \rtsmarmone (153~mJ) and 
$52\%$ for \rtsmarmtwo (221~mJ).

\subsection{Trade-off}
The results indicate that no single scheme simultaneously minimizes both secure storage and attestation runtime across all block sizes $B$, highlighting an inherent trade-off between these two objectives.

For small block sizes ($B \leq 256$~bytes), SMARM incurs 
substantial storage overhead ($2$-$13$~KB). Although 
\rtsmarmone reduces storage, its bitmap-based selection 
incurs significant runtime and energy overhead at large~$n$: 
at $B = 64$~bytes, \rtsmarmone consumes 2{,}650~mJ 
(13.9$\times$ SMARM). In this range, \rtsmarmtwo is the 
most suitable option, achieving larger storage reduction 
while maintaining lower energy consumption (1{,}243~mJ 
vs.\ 2{,}650~mJ at $B = 64$~bytes).

For medium block sizes (512--1024~bytes), \rtsmarmone 
provides the best balance across all three dimensions. 
At $B = 1024$~bytes, it achieves a $9\times$ storage 
reduction, only $7\%$ runtime overhead (526~ms vs.\ 
490~ms), and only $5\%$ energy overhead (153~mJ vs.\ 
145~mJ) relative to SMARM---making it the most 
energy-efficient \rtsmarm variant at the recommended 
operating point.
\rtsmarmtwo is also viable at this point, offering larger 
storage savings ($36\times$) at the cost of higher runtime 
(726~ms) and energy overhead (221~mJ, $52\%$ above SMARM).

For large block sizes ($B \geq 2048$~bytes), both variants 
have comparable storage, runtime, and energy results. 
However, this regime significantly reduces FAR 
below~1.0 (Section~\ref{sec:evaluation}), making it 
unsuitable for real-time workloads.

\subsection{Discussion}
\label{sec:resmarm-discussion}

Our results suggest that the choice among SMARM, \rtsmarmone, and \rtsmarmtwo depends not only on block size and runtime/energy overhead, but also on the available secure storage size that depends on how Secure RAM is allocated alongside the attestation service.

When Secure RAM is dedicated primarily to attestation, such as in attestation-focused architectures~\cite{smart,hydra}, baseline SMARM is the most attractive option.
In these settings, the permutation array can reside entirely in Secure RAM without competing with other services, making SMARM the simplest and most runtime-efficient design.

When Secure RAM must be shared with additional platform 
secrets~\cite{silabs-secure-key-storage}, \rtsmarmone 
becomes a more practical choice. At $B \geq 512$~bytes, 
it achieves a $9\times$ storage reduction while incurring 
only $5$--$26\%$ runtime and energy overhead relative to 
SMARM, making it the most balanced option for 
resource-constrained deployments.

For example, on the STM32 Nucleo-144 platform, selecting $B \geq 512$~bytes allows the bitmap to remain sufficiently small while preserving space for co-resident secure assets and maintaining low runtime overhead.

In more constrained environments where Secure RAM is heavily contended, \rtsmarmtwo is often preferable.
This scenario arises when the Secure world simultaneously hosts multiple security-critical services, such as confidential deep-learning inference~\cite{tslices} or RTOS-aware control-flow integrity monitoring~\cite{tzmcfi}.
Because \rtsmarmtwo requires only a constant-size storage footprint, it minimizes Secure RAM usage and enables coexistence with other secure services.

Finally, we observe that neither \rtsmarmone nor \rtsmarmtwo changes FAR behavior across different block sizes $B$.
To validate this empirically, we measured FAR for both variants using the same methodology described in Section~\ref{sec:evaluation} under low/medium/high real-time workloads.
The results closely match the baseline SMARM measurements reported in Section~\ref{sec:evaluation}, indicating that the proposed modifications influence storage and runtime overhead but do not alter real-time interference characteristics for a given block size.

\section{Related Work}
\label{sec:related}

\begin{table}[!ht]
  \centering
  \scriptsize
  \setlength{\tabcolsep}{3pt}
  \begin{tabular}{lcccc}
    \toprule
    \textbf{Work} &
    \textbf{RTOS} &
    \textbf{RT eval.} &
    \textbf{No extra HW} &
    \textbf{Secure Storage} \\
    \midrule
    SMARM~\cite{smarm}
        & \texttimes & \texttimes
        & \checkmark & $O(n \log n)$ \\
    PRoM~\cite{prom}
        & \texttimes & \texttimes
        & \checkmark & $O(n \log n)$ \\
    HAtt~\cite{hatt}
        & \texttimes & \texttimes
        & \texttimes & $O(n \log n)$ \\
    ITERATOR~\cite{iterator}
        & \texttimes & \texttimes
        & \checkmark & $O(n)$ \\
    \midrule
    \textbf{\rtsmarmone\ (ours)}
        & \checkmark & \checkmark
        & \checkmark & $O(n)$ \\
    \textbf{\rtsmarmtwo\ (ours)}
        & \checkmark & \checkmark
        & \checkmark & $O(1)$ \\
    \bottomrule
  \end{tabular}
  \caption{Position of our work w.r.t. prior work in shuffled-measurement remote attestation schemes.}
  \label{tab:related}
\end{table}

\subsection{Shuffled Measurement Techniques in Remote Attestation}
\label{sec:related-smarm}

Several works have extended SMARM~\cite{smarm} core idea of measuring memory blocks in a shuffled order, each addressing specific challenges in attestation for IoT devices.

\textbf{PRoM}~\cite{prom} applies shuffled measurement
to multi-core IoT devices. Rather than measuring all
blocks sequentially on a single core, PRoM distributes
the attestation across multiple cores, reducing the
overall attestation time and improving availability
compared to SMARM. However, PRoM must inherit SMARM's need for non-interruptible block measurement. 
PRoM also has not been evaluated under real-time workloads, leaving its impact on time-sensitive tasks unknown.

\textbf{HAtt}~\cite{hatt} improves availability by
using Physical Unclonable Functions (PUFs) to secure
attestation secrets and adopts SMARM's randomized
measurement approach. 
Similar to SMARM, HAtt mandates per-block atomicity.
Furthermore, HAtt introduces a hardware dependency on PUFs that may not be available on standard TrustZone-M MCUs, limiting its deployability on off-the-shelf embedded platforms.

\textbf{ITERATOR}~\cite{iterator} adopts a different strategy.
In an offline phase, it constructs a Cuckoo filter over all memory blocks to represent the valid software state.
During runtime, ITERATOR incrementally verifies memory through two operations: \textsc{Attest}, which validates a newly selected block against the filter, and \textsc{Check}, which re-validates $k$ randomly chosen previously measured blocks to defend against relocation malware.
ITERATOR still requires interrupts to remain disabled during each consecutive \textsc{Attest}-\textsc{Check} sequence.
As such, non-interruptible execution is not eliminated in ITERATOR.

While aforementioned studies build upon SMARM's core idea, they have neither been implemented on RTOS platforms nor evaluated under real-time workloads, leaving their practicality in realistic real-time deployments questionable.
Moreover, they inherit SMARM's secure-storage limitation, which can become prohibitive on memory-constrained devices.
Although ITERATOR reduces storage complexity to $O(n)$ through Cuckoo filters, its actual overhead depends on parameters such as the target false-positive rate and load factor~\cite{fan2014cuckoo}, resulting in a larger constant cost in practice.
In contrast, \rtsmarmone requires exactly $n$ bits of storage with no additional asymptotic constant overhead, while \rtsmarmtwo eliminates storage dependence on $n$ entirely.
Table~\ref{tab:related} summarizes the differences between our work and prior work in shuffled measurements.

\subsection{Other Approaches to Realize Interruptible Attestation}
\label{sec:related-ra}

Beyond shuffled measurements, several alternative approaches have been proposed to relax the atomicity requirement imposed by traditional RA schemes~\cite{ra_survey,hydra,smart,nunes2024toward,sancus2}.

Self-measurement techniques~\cite{erasmus,carpent2018reconciling,iterator} allow the prover to initiate attestation independently of the verifier, enabling measurements to be scheduled at times that minimize interference with real-time tasks.
Other works~\cite{iscflat,asap,pistis} permit interrupts during attestation by additionally measuring interrupt vector tables (IVTs), based on the intuition that malicious interrupt-driven behavior will be reflected in the attestation evidence.

A different line of research relies on hardware support to reduce attestation latency or measurement disruption.
Examples include hardware-assisted memory monitoring or evidence generation~\cite{li2023interruptible,lofat,litehax,atrium}, which can minimize or eliminate non-interruptible execution windows.
Finally, some approaches~\cite{casu,rata} focus on preventing malware persistence by enforcing strict write protection policies that block unauthorized memory modifications.

Complementary to shuffled measurement, recent work shows that long
atomic Secure-world services on TrustZone-M can suppress RTOS SysTick
handling and distort Non-Secure timekeeping~\cite{joianeto2026timekeeping}.
That line of work targets tick reconstruction during atomic secure
execution, whereas our FAR metric quantifies frequency-level interference
from per-block atomic \emph{memory} hashing in SMARM; the two concerns
are orthogonal but both arise when RTOS tasks coexist with TrustZone-M
attestation.

These alternatives are complementary to shuffled measurements and address different aspects of interruptible attestation.
Studying how shuffled measurements can be combined with these techniques to jointly improve security and real-time compatibility is an interesting direction for future work.

\section{Conclusion}
\label{sec:conclusion}

This paper presented the first systematic study of shuffled-memory remote attestation in real-time RTOS-based IoT systems. We realized SMARM on commodity TrustZone-M hardware running FreeRTOS and Zephyr, replacing its original microkernel-based design while preserving the same security guarantees.

Our evaluation revealed a fundamental limitation of SMARM: small block sizes preserve real-time availability but require substantial secure storage, restricting deployment on memory-constrained devices. To address 
this limitation, we proposed \rtsmarm with two variants that reduce secure-storage requirements through different design trade-offs. Experimental results show that both variants significantly lower secure-storage overhead while preserving SMARM's real-time compatibility, at 
the cost of additional attestation runtime and energy overhead. At the recommended operating point ($B = 1024$~bytes), \rtsmarmone incurs only $7\%$ runtime and $5\%$ energy overhead over SMARM, while \rtsmarmtwo offers larger storage savings at the cost of $48\%$ runtime and $52\%$ energy overhead.Overall, \rtsmarm provides a practical foundation for 
deploying shuffled remote attestation on real-time platforms with limited secure storage.







\appendices

\section{Formal Security Analysis}
\label{sec:formal-security}
 
This section provides a formal security analysis of our TrustZone-M-based SMARM realization in Section~\ref{sec:porting} and the two SMARM+ variants in Section~\ref{sec:resmarm-design}.
Our goal is to show that each construction satisfies the three security requirements (S1)-(S3) defined in Section~\ref{sec:background} under standard cryptographic and hardware assumptions.
 
We refer to $\lambda$ as the security parameter, $\mathit{negl}(\lambda)$ for
any negligible function, and PPT for probabilistic polynomial-time.
All algorithms receive $1^{\lambda}$ as an implicit input.
We denote by $[n]$ the set $\{0,\dots,n{-}1\}$ of memory block indices.
 
\subsection{Preliminaries}
\label{sec:prelim}
 
\begin{definition}[Pseudorandom Function (PRF)]
\label{def:prf}
A keyed function $F : \{0,1\}^{\lambda} \times \{0,1\}^{*} \to \{0,1\}^{\lambda}$
is a \emph{pseudorandom function} (PRF) if for every PPT distinguisher
$\mathcal{D}$,
\[
\begin{aligned}
  \Adv^{\mathrm{PRF}}_{\mathcal{D}}(\lambda)
  &=
  \left|
    \Pr\!\left[
      \mathcal{D}^{F(k,\cdot)}(1^{\lambda}) = 1
    \right]
    -
    \Pr\!\left[
      \mathcal{D}^{R(\cdot)}(1^{\lambda}) = 1
    \right]
  \right| \\
  &\leq \mathit{negl}(\lambda)
\end{aligned}
\]
where $k \xleftarrow{\$} \{0,1\}^{\lambda}$ and $R$ is a uniformly random
function.
\end{definition}
 
\begin{definition}[Pseudorandom Generator (PRNG)]
\label{def:prng}
A deterministic function $G : \{0,1\}^{\lambda} \to \{0,1\}^{\ell(\lambda)}$
with $\ell(\lambda) > \lambda$ is a \emph{pseudorandom generator} if for
every PPT distinguisher $\mathcal{D}$,
\[
\begin{aligned}
  \Adv^{\mathrm{PRNG}}_{\mathcal{D}}(\lambda)
  &=
  \left|
    \Pr\!\left[\mathcal{D}(G(s)) = 1\right]
    -
    \Pr\!\left[\mathcal{D}(r) = 1\right]
  \right| \\
  &\leq \mathit{negl}(\lambda)
\end{aligned}
\]
where $s \xleftarrow{\$} \{0,1\}^{\lambda}$ and
$r \xleftarrow{\$} \{0,1\}^{\ell(\lambda)}$.
\end{definition}
 
\begin{definition}[Pseudorandom Permutation (PRP)]
\label{def:prp}
A keyed family of bijections
$E : \{0,1\}^{\lambda} \times [n] \to [n]$
is a \emph{pseudorandom permutation} over domain $[n]$ if for every
PPT distinguisher $\mathcal{D}$,
\[
\begin{aligned}
  \Adv^{\mathrm{PRP}}_{\mathcal{D}}(\lambda)
  &=
  \left|
    \Pr\!\left[\mathcal{D}^{E(k,\cdot)}(1^{\lambda}) = 1\right]
    -
    \Pr\!\left[\mathcal{D}^{\pi(\cdot)}(1^{\lambda}) = 1\right]
  \right|\\
  &\leq \mathit{negl}(\lambda)
\end{aligned}
\]
where $k \xleftarrow{\$} \{0,1\}^{\lambda}$ and $\pi$ is a uniformly random
permutation over $[n]$.
\end{definition}
 
Our proofs rely on the following standard cryptographic assumptions:
 
\begin{assumption}[SHA-256 as PRF]
\label{asm:sha256-prf}
Keyed hash is a PRF (Definition~\ref{def:prf})
when keyed over $\{0,1\}^{256}$.
This is the standard assumption~\cite{bellare1996new} and is
widely accepted for SHA-256 as the underlying cryptographic hash function.
\end{assumption}
 
\begin{assumption}[PRNG security]
\label{asm:prng}
The PRNG used in \rtsmarmone, seeded with a uniformly random value, produces
output that is computationally indistinguishable from a uniform random string
(Definition~\ref{def:prng}).
\end{assumption}
 
\begin{assumption}[FFX-Speck as PRP]
\label{asm:ffx-prp}
The FFX construction~\cite{bellare2010ffx} instantiated with the Speck block
cipher~\cite{speck} is a PRP (Definition~\ref{def:prp}) over
domain $[n]$, provided that no adversary is given access to an
encryption or decryption oracle for the underlying FFX instance.
\end{assumption}

We also assume the following assumptions about TrustZone-M.
 
\begin{assumption}[TrustZone-M Isolation]
\label{asm:tz-isolation}
No code executing in the Non-Secure world can read or write memory
regions assigned to the Secure world by the SAU/IDAU configuration,
even under full-software compromise of Non-Secure world.
\end{assumption}
 
\begin{assumption}[TrustZone-M Entry Integrity]
\label{asm:tz-entry}
Secure-world functions are invocable from the Non-Secure world only via
explicitly designated NSC entry points.
Mid-function entry and early exit are prevented by the secure state transition mechanism enforced in hardware.
\end{assumption}

\subsection{Security Game}
\label{sec:game}
 
We define a unified \emph{measurement-order indistinguishability} game that
captures S3 for all three constructions.
S1 and S2 are architectural properties enforced by TrustZone-M hardware
and are addressed separately in Section~\ref{sec:tzm-security}.
 
\begin{definition}[Measurement-Order Indistinguishability Game
  $\mathsf{MOI}^{\Pi}_{\mathcal{A}}(\lambda)$]
\label{def:moi-game}
Let $\Pi$ be a memory measurement scheme over $n$ blocks.
The game $\mathsf{MOI}^{\Pi}_{\mathcal{A}}(\lambda)$ between a challenger
$\mathcal{C}$ and a PPT adversary $\mathcal{A}$ proceeds as follows:
\begin{enumerate}
  \item \textbf{Setup.}
    $\mathcal{C}$ samples a secret key
    $k \xleftarrow{\$} \{0,1\}^{\lambda}$ and initializes $\Pi$ with $k$.
  \item \textbf{Challenge.}
    The verifier sends a fresh challenge
    $\mathit{chal} \xleftarrow{\$} \{0,1\}^{\lambda}$;
    $\mathcal{C}$ makes $\mathit{chal}$ available to $\mathcal{A}$.
  \item \textbf{Adversary capabilities.}
    $\mathcal{A}$ controls all Non-Secure-world software.
    Concretely, $\mathcal{A}$:
    \begin{itemize}
      \item knows $n$, $\mathit{chal}$, and the complete specification of $\Pi$;
      \item observes the wall-clock timestamps $t_1,\dots,t_n$ at which interrupts are re-enabled after each block measurement (KFV capability~\cite{smarm});
      \item may read or write any Non-Secure RAM address at any time.
    \end{itemize}
  \item \textbf{Output.}
    $\mathcal{A}$ outputs a predicted permutation
    $\hat{\sigma} \in \mathrm{Perm}([n])$.
  \item \textbf{Win condition.}
    Let $\sigma^{\Pi}$ be the actual measurement order produced by $\Pi$
    on inputs $(k, \mathit{chal})$.
    $\mathcal{A}$ wins if $\hat{\sigma} = \sigma^{\Pi}$.
\end{enumerate}
The advantage of $\mathcal{A}$ against $\Pi$ is
\[
  \Adv^{\mathrm{MOI}}_{\mathcal{A},\Pi}(\lambda)
  \;=\; \Pr\!\left[\mathcal{A}\ \text{wins}\ \mathsf{MOI}^{\Pi}_{\mathcal{A}}(\lambda)\right]
  - \frac{1}{n!},
\]
where $\frac{1}{n!}$ is the probability of a correct random guess over
$\mathrm{Perm}([n])$.
\end{definition}
 
We say $\Pi$ satisfies \emph{measurement-order indistinguishability} if
$\Adv^{\mathrm{MOI}}_{\mathcal{A},\Pi}(\lambda) \leq \mathit{negl}(\lambda)$
for every PPT $\mathcal{A}$.
 
\subsection{Security of the TrustZone-M-based \smarm}
\label{sec:tzm-security}

\begin{theorem}[Security of RTOS-based SMARM]
\label{thm:tzm}
The TrustZone-M-based SMARM realization (SMARM-TZM) described in Section~\ref{sec:porting} satisfies: (S1) Controlled Execution,
(S2) Key Secrecy, and (S3) Private Measurement Order.
\end{theorem}

\begin{proof}

\textbf{(S1) Controlled Execution.}
The SMARM Service is implemented entirely in the Secure world and is accessible from the Non-Secure world only through a designated NSC entry point.
By Assumption~\ref{asm:tz-entry}, the hardware prevents both mid-code entry into the Secure world and premature return to the Non-Secure world.
The adversary can influence Secure-world execution only when interrupts are enabled after each block measurement, i.e., Line~7 of Algorithm~\ref{alg:smarm_baseline}.
At that point, an interrupt may temporarily transfer control away from the SMARM Service, but upon resumption execution returns to the same interrupted location; it cannot redirect control to any other Secure-world code location.
Moreover, interrupts are disabled during each block measurement, so the adversary cannot interrupt the execution of an individual block.
Therefore, Secure-world execution is confined to the intended SMARM control flow, satisfying~(S1).
 
\medskip
\noindent\textbf{(S2) Key Secrecy.}
The attestation key $k$ is provisioned once and stored in Secure flash,
which the SAU/IDAU configuration marks as Secure-only.
At runtime, all $k$-derived material resides exclusively in Secure RAM.
By Assumption~\ref{asm:tz-isolation}, no Non-Secure code, including a
fully compromised RTOS or any application task, can read or write these locations.
Hence, $k$ is never exposed to the adversary, satisfying (S2).
 
\medskip
\noindent\textbf{(S3) Private Measurement Order.}
We show that no PPT adversary $\mathcal{A}$ can predict the permutation
$\sigma$ used in a given attestation instance with non-negligible advantage.
 
Recall $\sigma$ is generated by $\mathrm{FisherYates}(\mathit{seed},\, [0,\dots,n-1])$, where $\mathit{seed} = \mathrm{Hash}(\mathit{chal},\,k)$.
After generation, $\sigma$ is stored in Secure RAM and never written to
Non-Secure RAM.
By Assumption~\ref{asm:tz-isolation}, $\mathcal{A}$ cannot read $\sigma$ directly.
 
We argue that $\mathcal{A}$ cannot compute $\sigma$ from $\mathit{chal}$
alone.
Suppose, for contradiction, that $\mathcal{A}$ could predict $\sigma$ with
non-negligible advantage.
Since $\sigma$ is a deterministic function of $\mathit{seed}$, this implies
$\mathcal{A}$ can recover $\mathit{seed}$.
But $\mathit{seed} = \mathrm{Hash}(\mathit{chal},\,k)$, and by
Assumption~\ref{asm:sha256-prf} (SHA-256 as PRF), an adversary that knows
$\mathit{chal}$ but not $k$ cannot distinguish $\mathit{seed}$ from a
uniformly random string with non-negligible probability.
This contradicts our supposition; hence
$\Adv^{\mathrm{MOI}}_{\mathcal{A}, \text{SMARM-TZM}}(\lambda) \leq \mathit{negl}(\lambda)$,
satisfying (S3).
 
\end{proof}
 
\subsection{Security of \rtsmarmone}
\label{sec:proof-prng}
 
\begin{theorem}[\rtsmarmone Security]
\label{thm:prng}
SMARM+PRNG satisfies (S1)-(S3).
Specifically, let $\Pi$ be \rtsmarmone. For every PPT adversary $\mathcal{A}$,
\[
  \Adv^{\mathrm{MOI}}_{\mathcal{A},\Pi}(\lambda) \;\leq\; \mathit{negl}(\lambda),
\]
\end{theorem}
 
\begin{proof}
(S1) and (S2) are inherited from Theorem~\ref{thm:tzm} unchanged:
SMARM+PRNG retains the same TrustZone-M architecture, the same NSC entry
point, and the same Secure-RAM storage for $k$ and all derived material.
 
We prove (S3) via the following games:
 
\medskip
\noindent\textbf{Game $G_0$: Real scheme.}
This is the $\mathsf{MOI}^\Pi_{\mathcal{A}}(\lambda)$ game
(Definition~\ref{def:moi-game}) with $\Pi$ being \rtsmarmone, where the measurement is performed in the following order:
\begin{align*}
  \mathit{seed}    &= \mathrm{SHA256}(\mathit{chal} \,\|\, k) \\
  \text{PRNG state} &= \mathrm{PRNG.Init}(\mathit{seed}) \\
  r_i              &= \mathrm{PRNG.RandInt}(1,\, n-i)
                      \ \text{for } i = 0,\dots,n-1 \\
  \mathit{idx}_i   &= \mathrm{FindKthZero}(\mathit{used},\, r_i)
\end{align*}
The sequence $(\mathit{idx}_0,\dots,\mathit{idx}_{n-1})$ constitutes the
measurement order $\sigma^\Pi$.
Let $p_0 = \Pr[\mathcal{A}\ \text{wins in } G_0]$.
 
\medskip
\noindent\textbf{Game $G_1$: Replace seed with uniform random.}
We replace $\mathit{seed} = \mathrm{SHA256}(\mathit{chal} \,\|\, k)$ with
$\mathit{seed} \xleftarrow{\$} \{0,1\}^{256}$.
All other steps are identical.
Let $p_1 = \Pr[\mathcal{A}\ \text{wins in } G_1]$.
 
\begin{lemma}
\label{lem:g0-g1}
$|p_0 - p_1| \leq \Adv^{\mathrm{PRF}}_{\mathcal{B}_1}(\lambda)$.
\end{lemma}
\begin{proof}
Suppose $|p_0 - p_1| > negl(\lambda)$.
Then, we can construct a PRF distinguisher $\mathcal{B}_1$ as follows:
$\mathcal{B}_1$ receives oracle access to either $F = \mathrm{SHA256}(k,\cdot)$
(real PRF) or a truly random function $\mathcal{R}$.
It queries its oracle on $\mathit{chal}$ to obtain the seed, runs the rest
of \rtsmarmone honestly, and outputs whatever $\mathcal{A}$ outputs.
When $\mathcal{B}_1$ oracle is the real PRF, it simulates $G_0$ perfectly;
when it is $\mathcal{R}$, it simulates $G_1$ perfectly.
Hence $\Adv^{\mathrm{PRF}}_{\mathcal{B}_1}(\lambda) \ge |p_0 - p_1| > negl(\lambda)$,
contradicting Assumption~\ref{asm:sha256-prf}.
\end{proof}
 
\medskip
\noindent\textbf{Game $G_2$: Replace PRNG output with uniform random.}
Given a uniformly random seed (as in $G_1$), we further replace the PRNG
outputs $r_0,\dots,r_{n-1}$ with truly uniform random integers
$r_i \xleftarrow{\$} [1, n-i]$ for each $i$.
Let $p_2 = \Pr[\mathcal{A}\ \text{wins in } G_2]$.
 
\begin{lemma}
\label{lem:g1-g2}
$|p_1 - p_2| \leq \Adv^{\mathrm{PRNG}}_{\mathcal{B}_2}(\lambda)$.
\end{lemma}
\begin{proof}
A PRNG seeded with a uniformly random value (as in $G_1$) is computationally indistinguishable from a uniform random string by Assumption~\ref{asm:prng}.
Any PPT distinguisher $\mathcal{B}_2$ that separates $G_1$ from $G_2$
directly breaks Assumption~\ref{asm:prng}.
\end{proof}
 
In $G_2$, each $r_i$ is a fresh uniform sample from $[1, n-i]$,
independent of $\mathcal{A}$'s view.
The mapping $\mathrm{FindKthZero}(\mathit{used},\, r_i)$ is a bijection
between $[1, n-i]$ and the set of remaining (unmeasured) block indices at
step $i$.
Hence, the resulting sequence
$(\mathit{idx}_0,\dots,\mathit{idx}_{n-1})$ is a uniformly random
permutation over $[n]$: at each step $i$, the selected block is chosen
uniformly at random from the $n - i$ remaining blocks, and this process
produces the uniform distribution over $\mathrm{Perm}([n])$.
As such, $\mathcal{A}$'s best strategy in $G_2$ is random guessing:
$p_2 = \frac{1}{n!}$, resulting in:
\[
\begin{aligned}
  \Adv^{\mathrm{MOI}}_{\mathcal{A},\Pi}(\lambda)
  &= p_0 - \tfrac{1}{n!}
  \leq |p_0 - p_1| + |p_1 - p_2| + (p_2 - \tfrac{1}{n!})\\
  &\leq \Adv^{\mathrm{PRF}}_{\mathcal{B}_1}(\lambda)
       + \Adv^{\mathrm{PRG}}_{\mathcal{B}_2}(\lambda)\\
  &\leq \mathit{negl}(\lambda).
\end{aligned}
\]
This proves Theorem~\ref{thm:prng}.
 
\end{proof}
 
\subsection{Security of \rtsmarmtwo}
\label{sec:proof-fpe}
 
\begin{theorem}[\rtsmarmtwo Security]
\label{thm:fpe}
\rtsmarmtwo satisfies (S1)-(S3).
Specifically, let $\Pi$ be \rtsmarmtwo.
For every PPT adversary $\mathcal{A}$,
\[
  \Adv^{\mathrm{MOI}}_{\mathcal{A},\Pi}(\lambda) \;\leq\; \mathit{negl}(\lambda),
\]
\end{theorem}

\begin{proof}
(S1) and (S2) are inherited from Theorem~\ref{thm:tzm} by the same
argument as in the proof of Theorem~\ref{thm:prng}.
We prove (S3) via the following games:
 
\medskip
\noindent\textbf{Game $G_0$: Real scheme.}
The measurement order in \rtsmarmtwo is:
\begin{align*}
  (\mathit{key},\, \mathit{tweak})
      &= \mathrm{KDF}(k,\, \mathit{chal}), \\
  \mathit{idx}_i
      &= \mathrm{FFX.Enc}_{\mathit{key},\,\mathit{tweak}}(i)
         \ \text{for } i = 0,\dots,n-1.
\end{align*}
Since $\mathrm{FFX.Enc}$ is a bijection over $[n]$, the sequence
$(\mathit{idx}_0,\dots,\mathit{idx}_{n-1})$ is $\sigma^\Pi$.
Let $p_0 = \Pr[\mathcal{A}\ \text{wins in }G_0]$.
 
\medskip
\noindent\textbf{Game $G_1$: Replace KDF output with uniform random key.}
We replace $(\mathit{key},\mathit{tweak}) = \mathrm{KDF}(k,\mathit{chal})$
with $(\mathit{key},\mathit{tweak}) \xleftarrow{\$} \{0,1\}^{2\lambda}$.
All other steps are identical.
Let $p_1 = \Pr[\mathcal{A}\ \text{wins in }G_1]$.
 
\begin{lemma}
\label{lem:fpe-g0-g1}
$|p_0 - p_1| \leq \Adv^{\mathrm{PRF}}_{\mathcal{B}_1}(\lambda)$.
\end{lemma}
\begin{proof}
We prove by reduction to the PRF security in Assumption~\ref{asm:sha256-prf}.
Suppose $|p_0 - p_1|$ is non-negligible.
We show how to construct a PPT distinguisher $\mathcal{B}_1$ against the PRF security of
$\mathrm{KDF}(k,\cdot)$.

$\mathcal{B}_1$ is given access to an oracle function
which is either $\mathrm{KDF}(k,\cdot)$ for a secret key $k$ or a uniformly random
function with the same output length.
To simulate the game, $\mathcal{B}_1$ queries $O$ on
$\mathit{chal}$ and parses the response as
$(\mathit{key},\mathit{tweak})$.
It then uses this pair as the FFX key material and runs the remainder of
$G_0$ exactly as specified for $\mathcal{A}$.

When the oracle is $\mathrm{KDF}(k,\cdot)$, the simulation is identical to $G_0$.
If the oracle is a uniformly random function, then
$(\mathit{key},\mathit{tweak})$ is uniform and independent, so the simulation is
identical to $G_1$.
Therefore, $\mathcal{B}1$ distinguishes the PRF oracle from a random function
with non-negligible advantage $|p_0-p_1|$, breaking Assumption~\ref{asm:sha256-prf} (since $\mathrm{KDF}(k,\cdot)$ is instantiated as
$\mathrm{HMAC\text{-}SHA256}(k,\cdot)$ and is assumed to be a PRF by
Assumption~\ref{asm:sha256-prf})
\end{proof}

\medskip
\noindent\textbf{Game $G_2$: Replace FFX with a truly random permutation.}
Given a uniformly random FFX key (as in $G_1$), we replace
$\mathrm{FFX.Enc}_{\mathit{ffx\_key},\mathit{tweak}}$ with a truly random
permutation $\pi \xleftarrow{\$} \mathrm{Perm}([n])$.
Let $p_2 = \Pr[\mathcal{A}\ \text{wins in }G_2]$.
 
\begin{lemma}
\label{lem:fpe-g1-g2}
$|p_1 - p_2| \leq \Adv^{\mathrm{PRP}}_{\mathcal{B}_2}(\lambda)$.
\end{lemma}
\begin{proof}
We prove by reduction to the FFX security in Assumption~\ref{asm:ffx-prp}.
Suppose $|p_1 - p_2|$ is non-negligible.
We show how to construct a PRP distinguisher $\mathcal{B}_2$ with oracle access to
either $\mathrm{FFX.Enc}_{k'}$ (real PRP) or a truly random
permutation $\pi$ over $[n]$.
$\mathcal{B}_2$ runs $\mathcal{A}$ using its oracle as the FFX instantiation
(querying on $i = 0,\dots,n-1$) and forwards $\mathcal{A}$'s output as
its own distinguishing bit.
When the oracle is the real PRP, this simulates $G_1$ perfectly;
when it is $\pi$, it simulates $G_2$ perfectly.
Hence this makes $\Adv^{\mathrm{PRP}}_{\mathcal{B}_2}(\lambda)$ also non-negligble,
contradicting Assumption~\ref{asm:ffx-prp}.
\end{proof}
 
In $G_2$, the measurement order is a uniformly random permutation
$\pi \xleftarrow{\$} \mathrm{Perm}([n])$, independent of everything the
adversary knows.
Hence, $p_2 = \frac{1}{n!}$, resulting in:
\[
\begin{aligned}
  \Adv^{\mathrm{MOI}}_{\mathcal{A},\Pi}(\lambda)
  &= p_0 - \tfrac{1}{n!}
  \leq |p_0 - p_1| + |p_1 - p_2| + (p_2 - \tfrac{1}{n!})\\
  &\leq \Adv^{\mathrm{PRF}}_{\mathcal{B}_1}(\lambda)
       + \Adv^{\mathrm{PRP}}_{\mathcal{B}_2}(\lambda)\\
  &\leq \mathit{negl}(\lambda).
\end{aligned}
\]
This proves Theorem~\ref{thm:fpe}.

\end{proof}

\section*{Acknowledgment}
The authors used Claude and ChatGPT as an AI assistant to support drafting and editing of portions of this manuscript, including literature review text, section organization, proof assistant, and language refinement. All technical content, experimental results, and conclusions were verified and approved by the authors

\bibliographystyle{ieeetr}
\bibliography{reference}


\end{document}